\documentclass[12pt,a4paper,reqno]{amsart}

\usepackage{amsmath,amsfonts}
\usepackage{latexsym}
\usepackage{epic}
\usepackage{graphicx}
\graphicspath{{}}
\usepackage{ifthen}
\usepackage{colortbl}
\usepackage{xcolor}

\usepackage[foot]{amsaddr}

\usepackage[lof,lot]{etoc}

\usepackage{tikz,graphics,multicol}

\usepackage{algorithmic}
\usepackage{float}
\usetikzlibrary{calc,snakes} 

\newlength\tindent
\usepackage{nameref}

\newcommand{\pr}[2][]{\ensuremath{\mathrm{Pr}_{#1}\!\left(#2\right)}}

\def \R{\mathbb{R}}
\def \union{\cup}

\newcommand{\given}{\ensuremath{\ |\ }}
\let\oldvec\vec
\renewcommand{\vec}[1]{\boldsymbol{#1}}
\newcommand{\geod}[2]{\ensuremath{\Gamma_{#1,#2}}}

\newcommand{\N}{\ensuremath{{N'}}}
\newcommand{\B}[2]{\ensuremath{B(#1,#2)}}
\newcommand{\BHVmetric}[2]{\ensuremath{d_{\text{BHV}}(#1,#2)}}
\newcommand{\splits}[1]{\ensuremath{\sigma(#1)}}

\newcommand{\BHV}[1]{\text{BHV}_{#1}}
\newcommand{\orthant}[1]{\mathcal{O}_{#1}}
\newcommand{\id}{\ensuremath{\textit{I}}}

\newcommand{\GGFdens}{\ensuremath{f_{\text{GGF}}}}
\newcommand{\GGF}{\ensuremath{\text{GGF}}}

\newcommand{\Wm}[2]{\ensuremath{W(#1,#2;m)}}
\newcommand{\varjm}[2]{\ensuremath{\tau_{#1,#2}}}
\newcommand{\varstar}{\ensuremath{\tau^*_{j-a,l,m}}}
\newcommand{\walkdens}{\ensuremath{f_{W}}}

\newcommand{\multfixedbridges}[2]{\oldvec{\vec{#1}}^{(#2)}}

\newcommand{\BB}[2]{\ensuremath{B(#1,#2)}}
\newcommand{\BBRPlus}[2]{\ensuremath{B_{\mathbb{R}_{\geq 0}^{N'}}(#1,#2)}}
\newcommand{\BR}[3]{\ensuremath{\phi(#3\given #1,#2)}}

\newcommand{\Bpath}[2]{\ensuremath{\mathcal{B}(#1,#2)}}

\newcommand{\BRPluspath}[2]{\ensuremath{\mathcal{B}_{\mathbb{R}_{\geq 0}^{N'}}(#1,#2)}}
\newcommand{\BBR}[2]{\ensuremath{\Phi(#1,#2)}}

\newcommand{\permGroup}{\ensuremath{G_{N'}}}
\newcommand{\logmle}{\hat{\ell}}

\DeclareMathOperator*{\argmin}{arg\,min}

\definecolor{DarkBlue}{rgb}{0.0, 0.08, 0.45}
\definecolor{BrickRed}{cmyk}{0, .89, .94, .28}

\renewcommand\subsubsection{\@startsection{subsubsection}{3}{\z@}%
                                   {-3.25ex \@plus -1ex \@minus -.2ex}%
                                   {1.5ex \@plus .2ex}%
                                   {\normalfont\bfseries}}
\makeatother

\newtheorem{result}{Result}[section]
\newtheorem{theorem}[result]{Theorem}
\newtheorem{lemma}[result]{Lemma}
\newtheoremstyle{defn}
  {3ex}
  {3ex}
  {\upshape}
  {}
  {\bfseries}
  {.}
  { }
  {}
\theoremstyle{defn}
\newtheorem{definition}[result]{Definition}
\newtheorem{algorithm}[result]{Algorithm}

\newtheorem*{remark*}{Remark}
\newtheorem*{algorithm*}{Algorithm}
\let\oldproofname=\proofname
\renewcommand{\proofname}{\rm\bf{\oldproofname}}
\newtheoremstyle{break}
  {3ex}
  {3ex}
  {\upshape}
  {}
  {\bfseries}
  {.}
  {\newline}
  {}
\theoremstyle{break}
\newtheorem{rmks}[result]{Remarks}
\newtheorem*{rmks*}{Remarks}
\newtheorem{exmps}[result]{Examples}
\newtheorem*{exmps*}{Examples}

\newenvironment{remarks*}[1][]{%
\begin{rmks*}[#1]\leavevmode\vspace{-\baselineskip}%
}{\end{rmks*}}

\newenvironment{examples*}[1][]{%
\begin{exmps*}[#1]\leavevmode\vspace{-\baselineskip}%
}{\end{exmps*}}
\newtheoremstyle{ex-sc}
  {3ex}
  {3ex}
  {\upshape}
  {}
  {\scshape}
  {.}
  {\newline}
  {}
\theoremstyle{ex-sc}
\newtheorem{excs}{Exercises}

\numberwithin{equation}{section}
\numberwithin{figure}{section}

\begin{document}

\title[Brownian motion in tree space]{Brownian motion, bridges and {B}ayesian inference in phylogenetic tree space}

\author{William M.~Woodman and Tom~M.~W.~Nye}
\email{w.m.woodman2@newcastle.ac.uk, tom.nye@newcastle.ac.uk}
\address{School of Mathematics and Statistics\\
        Newcastle University\\ Newcastle upon Tyne\\ NE1 7RU\\ UK}

\begin{abstract}
Billera-Holmes-Vogtmann (BHV) tree space is a geodesic metric space of edge-weighted phylogenetic trees with a fixed leaf set. 
Constructing parametric distributions on this space is challenging due to its non-Euclidean geometry and the intractability of normalizing constants. 
We address this by fitting Brownian motion transition kernels to tree-valued data via a non-Euclidean bridge construction. 
Each kernel is determined by a source tree \( x_0 \) (the Brownian motion's starting point) and a dispersion parameter \( t_0 \) (its duration). 
Observed trees are modelled as independent draws from the transition kernel defined by \( (x_0, t_0) \), analogous to a Gaussian model in Euclidean space. 
Brownian motion is approximated by an \( m \)-step random walk, with the parameter space augmented to include full sample paths. 
We develop a bridge algorithm to sample paths conditional on their endpoints, and introduce methods for sampling a Bayesian posterior for \( (x_0, t_0) \) and for marginal likelihood evaluation. 
This enables hypothesis testing for alternative source trees. 
The approach is validated on simulated data and applied to an experimental data set of yeast gene trees. 
These methods provide a foundation for future development of a wider class of probabilistic models of tree-valued data.
\end{abstract}

\maketitle

\section{Introduction}

Data sets consisting of trees arise in several contexts: for example
medical imaging of branching structures such as blood vessels \cite{ayd09,skw14} or lungs \cite{fera11a,fera11b}; and in molecular biology where phylogenetic analysis of aligned genomic sequences produces samples of evolutionary trees. 
Analysing such data is challenging, since the space containing the data is usually highly non-Euclidean. 
Spaces of trees are typically not vector spaces, nor manifolds, but combine combinatorial features, usually the discrete branching pattern or topology of each tree, with continuous features, for example edge lengths or the shape of edges on a tree in 3-dimensional space. 
Nonetheless, spaces of trees are often geodesic metric spaces \cite{bill01,speyer2004,fera11a,gav2016,lueg2024}, and this geometric structure can be exploited in order to analyse data. 
A common approach is least squares estimation, for example calculation of a Fr\'echet mean \cite{miller2015,brown2020} (which minimises the sum of squared distances to the data points) or first principal component \cite{nye11,fer13}. 
The least squares approach suffers from some disadvantages.  
First, it is not based fundamentally on probabilistic reasoning, so it can be difficult to assign a measure of uncertainty to the estimates~\cite{bard13,barden2018}.
Secondly, least-squares estimators have a tendancy to lie in low-dimensional strata in tree space \cite{hotz13}, meaning, for example, that the Fr\'echet mean for a collection of binary phylogenetic trees is very often not itself a strictly bifurcating tree -- a drawback in biological applications.

An alternative approach is to construct flexible parametric families of probability distributions and use these to develop models for which maximum likelihood or Bayesian inference can be performed. 
In this paper we consider a specific space of phylogenetic trees, the Billera-Holmes-Vogtmann (BHV) tree space \cite{bill01}; we construct parametric families of distributions which are transition kernels of Brownian motion; and we develop algorithms to sample from Bayesian posteriors in order to fit such distributions to samples of data. 
BHV tree space is a metric space consisting of all edge-weighted phylogenetic trees on a fixed set of taxa, with a unique geodesic between any pair of trees and globally non-positive curvature. 
These properties support convex optimization and ensure uniqueness of Fr\'echet means \cite{sturm2003}. 
Distance and geodesic computations are also tractable \cite{owen_fast_2011}, enabling practical implementations of statistical methods in the space.

However, specifying normalized density functions in BHV tree space is difficult, since even simple distributions have intractable normalizing constants. 
Consequently, given a random sample of trees, likelihood functions are usually impossible to construct and parameter inference is intractable. 
For example, the volume of a unit radius ball in BHV tree space varies with the location of the ball, and is very difficult to compute (see Section~\ref{sec:BHV}). 
As a result, a likelihood function cannot be evaluted for the family of uniform distributions on unit radius balls parametrized by the centre of each ball.

To address this, we simulate simple stochastic processes on tree space to construct distributions and perform inference. 
Given a point $x_0$ in tree space and dispersion parameter $t_0>0$, let $\B{x_0}{t_0}$ denote the transition kernel of Brownian motion from $x_0$ with duration $t_0$  \cite{nye2020random}. 
This distribution is analogous to a multivariate normal distribution in $\R^N$, and although the probability density function cannot be written down in closed form on BHV tree space, $\B{x_0}{t_0}$  can be approximated by a suitably-defined $m$-step random walk $\Wm{x_0}{t_0}$ \cite{nye2020random}.
We model a given data set of points $x_1,\ldots,x_n$ as i.i.d. samples from $\B{x_0}{t_0}$ and infer $x_0, t_0$ in a Bayesian framework. 
A key component is a bridge construction: we sample random walk paths between $x_0$ and each $x_i$ conditional on these end points, and combine this with Metropolis-Hastings Markov chain Monte Carlo (MCMC) to sample the posterior for $x_0,t_0$. 
Although kernels of Brownian motion were originally proposed as models in BHV tree space in \cite{nye14diffusion}, related probabilistic methods based on diffusions have been developed on Riemannian manifolds \cite{sommer2015anisotropic,sommer2017bridge,eltzner2023diffusion}, where the diffusion source $x_0$ is called a diffusion mean \cite{eltzner2023diffusion}.
Methods for constructing bridges on manifolds have been developed \cite{sommer2017bridge}, but constructing bridges in BHV tree space poses substantial additional challenges due to singularities in the space. 

Other parametric models on BHV tree space have been explored previously. 
In \cite{wey14}, distributions with probability density functions of the form 
\begin{equation}\label{equ:KDE}
f(x)\propto\exp -\frac{\BHVmetric{x}{x_0}^2}{2t_0}
\end{equation}
were proposed, where $x,x_0$ are points in BHV tree space, $d_{\text{BHV}}$ is the BHV metric, and $t_0>0$ is a dispersion parameter. 
Surprisingly, this distribution is not $\B{x_0}{t_0}$: analogs of different constructions of the Euclidean Gaussian yield distinct distributions in tree space.  
In \cite{wey14}, kernel density estimates were constructed by summing density functions of the form~$\eqref{equ:KDE}$ at each data point in a sample. 
This contrasts with the present paper, in which we fit a single `Gaussian' distribution to samples of points. 
The normalising constant of~$\eqref{equ:KDE}$ is very difficult to compute and varies with the parameters $x_0,t_0$. 
In~\cite{wey14} this was ignored, but in \cite{weyenberg2016normalizing} normalising constants were calculated via a rough approximation (see Section~\ref{sec:KDE} below).
Related kernel density estimates have also been developed in the so-called tropical tree space \cite{speyer2004,yoshida2022tropical}. 
More recently, log-concave distributions on BHV tree space have been proposed \cite{takazawa2024maximum}. 
The aim is similar to ours, in that they fit a single distribution to samples in BHV tree space. 
However, the methods in \cite{takazawa2024maximum} are not computationally feasible for more than a few taxa. 

The main contributions of this paper are as follows. 
In Section~\ref{sec:bridging} we describe bridge proposal algorithms and a MCMC sampler targeting the distribution of the $m$-step random walk conditional on the start and end points. 
Then, in Section~\ref{sec:inference}, we model samples of points in BHV tree space as i.i.d.~draws from the Brownian motion kernel $\B{x_0}{t_0}$. 
We describe a MCMC sampler which draws from the Bayesian posterior for $x_0,t_0$, and prove a consistency result. 
Finally, in Section~\ref{sec:marginalLikelihood} we specify algorithms for computing marginal likelihoods, effectively enabling computation of normalising constants for the Brownian motion kernels. 
The methodology is evaluated on simulated data (Section~\ref{sec:simulation}) and applied to a biological example (Section~\ref{sec:biologicalex}). 
In the basic form presented here, the methodology offers a novel Bayesian estimator of mean and dispersion summary statistics for a sample of phylogenetic trees.  
Beyond this, it provides a foundation for new probabilistic methods in tree space, such as hypothesis testing, regression and emulation.

\section{Background}\label{sec:background}
\subsection{BHV tree space}\label{sec:BHV}

In this section we fix notation and describe the geometry of the Billera-Holmes-Vogtmann tree space \cite{bill01}, emphasizing aspects relevant to our methodology. Since the stochastic processes we study avoid high co-dimensional singularities almost surely, full understanding of the details of BHV tree space is not required.

\textbf{Topological structure.} Let $N\geq 4$ be the number of taxa, and let $[N]=\{1,\ldots,N\}$ be the set of leaf labels. 
A phylogenetic tree on $[N]$ is an unrooted tree whose leaves are bijectively labelled by $[N]$, interior edges have strictly positive weight (also referred to as lengths), and no vertex has degree 2. 
Edges attached to leaves are called pendant edges; all others edges are called interior. 
Such trees contain at most $N-3$ interior edges; trees with fewer interior edges are called unresolved.  
The space $\BHV{N}$ is the set of all such trees, both resolved and unresolved, omitting pendant edge lengths as in \cite{bill01}. 
Our account considers unrooted trees, unlike \cite{bill01} where trees are rooted using an additional taxon. 
A rooted version of our methodology can be obtained by simply adding an additional taxon label to give a leaf set labelled $0,1,\ldots,N$, with taxon $0$ giving the position of the root. 

Cutting any internal edge in a tree induces a bipartition $\{A,A^c\}$ of $[N]$ where $A\subset [N]$ has between $2$ and $N-2$ elements and $A\union A^c=[N]$, $A\cap A^c=\emptyset$. 
Such a bipartition is called a split, and for any $x\in\BHV{N}$ we let $\splits{x}$ denote the set of interior splits in $x$, also called the topology of $x$. 
For any split $e$, let $x(e)$ be the edge length if $e \in \splits{x}$, and zero otherwise. 
There are  $M=2^{N-1}-(N+1)$ possible interior splits of $[N]$, and so fixing some ordering $e_1,\ldots,e_M$ of these splits, each $x \in \BHV{N}$ corresponds to a unique vector $\left(x(e_i)\right)_{i=1}^M\in\R^M$. 
Arbitrary sets of splits do not necessarily correspond to valid tree topologies: a compatibility condition must be met, since for example, the two splits $1,2|3,4,\ldots,N$ and $1,3|2,4,\ldots,N$ cannot both be contained in a tree. 
Thus, $\BHV{N}$ embeds as a subset of $\R^M$. 

\begin{figure}
\begin{center}
\includegraphics[width=0.8\textwidth,clip,trim={0.1cm 0cm 0.1cm 0.1cm}]{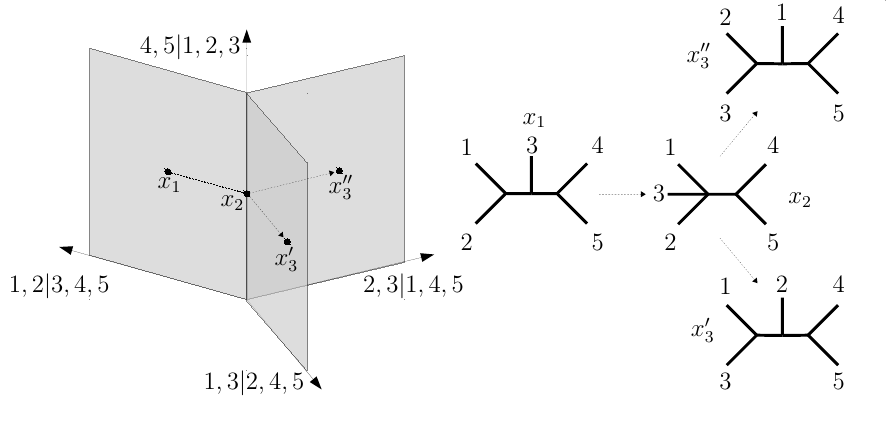}
\begin{caption}{
\label{fig:NNI}
Left: three orthants in $\BHV{5}$. The axes are labelled with the corresponding splits, and the position of a point in each orthant determines the length of the two corresponding internal edges. 
Right: trees corresponding to the points $x_1,x_2,x_3',x_3''$. 
Contracting split $1,2|3,4,5$ to length zero in tree $x_1$ yields the unresolved tree $x_2$. 
This tree can be resolved in two different ways from $x_2$, to give trees $x_3',x_3''$. 
}
\end{caption}
\end{center}
\end{figure}

The set of trees with a fixed fully resolved topology are parametrized by a positive orthant $\R^{N-3}_{>0}\subset\R^M$, and it can be shown that there are $1\times 3 \times 5\times\cdots\times(2N-5)$ such topologies. 
The points at the boundaries of the corresponding orthants, at which one or more split lengths have contracted to zero, correspond to unresolved trees. 
Each unresolved tree topology containing $0\leq s <N-3$ splits corresponds to an orthant $\R^s_{>0}\subset \R^M$ parametrising all trees $x\in\BHV{N}$ with that topology. 
The details of how these orthants meet, forming the stratification of $\BHV{N}$, is not vital for our methodology, but we note the following points. 
\begin{enumerate}
\item The origin of $\R^M$ corresponds to the star tree which contains the $N$ pendant edges but no interior edges. It lies in the closure of every orthant. 
\item Contracting a single interior edge in a fully resolved tree produces a degree-4 vertex. 
This vertex can be resolved by expanding out two alternative splits, as illustrated in Figure~\ref{fig:NNI}, and this operation is called Nearest Neighbour Interchange (NNI).
It follows every maximal orthant (a copy of $\R^{N-3}_{>0}$) is joined to two other maximal orthants at each of its codimension-1 boundaries. 
\item Orthants with $N-(k+3)$ splits are codimension-$k$ boundaries of maximal orthants, $k=1,\ldots, N-3$. 
\end{enumerate}

When $N=4$ each fully resolved tree contains a single internal split and there are three possibilities for this.
Then $\BHV{4}$ consists of the positive axes in $\R^3$, where the position along an axis gives the length of the corresponding split and the origin is the star tree. 
For $N=5$ each fully resolved tree contains two internal splits, and there are 10 possibilities for these. 
$\BHV{5}\subset \R^{10}$ contains 15 maximal orthants (copies of $\R^2_{>0}$) each joined to two neighbours along each codimension-1 boundary. 

\textbf{Metric structure.} For $x_1, x_2 \in \BHV{N}$ with the same topology, the distance $\BHVmetric{x_1}{x_2}$ is their Euclidean distance in $\R^M$. 
For trees with different topologies, we consider paths in $\BHV{N}\subset\R^M$ between $x_1$ and $x_2$ consisting of straight line segments within each orthant, and define the length of a path to be the sum of the lengths of these line segments. 
Then $\BHVmetric{x_1}{x_2}$ is the infimum of the length of these paths. 
This infimum is attained by a unique geodesic \cite{bill01}, computable in $O(N^4)$ time \cite{owen_fast_2011}.

Furthermore, it was shown that $\BHV{N}$ is a CAT$(0)$ space \cite{bill01}, a condition on the curvature of the space which implies a number of attractive geometrical properties, including existence and uniqueness of Fr\'echet means.  
For a sample $x_1,\ldots,x_n\in\BHV{N}$, the Fr\'echet sample mean defined by
\begin{equation}\label{equ:frechet}
\argmin_{x\in\BHV{N}} \frac{1}{n}\sum_{i=1}^n \BHVmetric{x}{x_i}^2.
\end{equation}
exists and is unique. 
This mean may lie in a lower-dimensional orthant even when all samples are resolved, and for some data, the mean remains unchanged under small perturbations of the data. 
This phenomenom is called \textit{stickiness} and has it important implications for the asymptotic theory of the Fr\'echet mean \cite{hotz13,bard13}.   

The geodesic $\Gamma_{x_1,x_2}$ between $x_1,x_2\in\BHV{N}$ deforms $x_1$ to $x_2$ by contracting and expanding edges. 
We will parametrize $\Gamma_{x_1,x_2}(t)$ with $t \in [0,1]$. 
For $x_1,x_2$ with different fully resolved topologies, $\Gamma_{x_1,x_2}$ will traverse singularities, namely regions with codimension $>1$, and in general, geodesics can traverse high-codimension strata over open subintervals of $[0,1]$. 

\begin{definition}
A geodesic $\Gamma_{x_1,x_2}$ is a \emph{cone path} if it consists of the straight line segments from $x_1$ to the origin, and from the origin to $x_2$. 
It is \emph{simple} if $x_1,x_2$ are fully resolved and the geodesic remains in codimension $\leq 1$ regions.
\end{definition} 

We equip $\BHV{N}$ with the Borel sigma algebra induced by its metric \cite{willis_confidence_2016}. 
On each orthant, this coincides with the usual Borel subsets of $\R^{N-3}_{>0}$. 
The reference Borel measure is obtained by summing Euclidean measures over maximal orthants.

\subsection{Random walks and Brownian motion in BHV tree space}

Brownian motion on $\BHV{N}$ was formally defined in \cite{nye2020random} as a Markov process $X(t)$ starting from $X(0) = x_0$, where $x_0$ is either fully resolved or in a codimension-1 orthant. 
Within a maximal orthant, $X(t)$ evolves as standard Brownian motion until it hits a codimension-1 boundary, at which point it moves continuously and uniformly at random to one of the three adjacent orthants. 
We write $\B{x_0}{t_0}$ for the distribution of $X(t_0)$ given $X(0)=x_0$.

For $N=4$, the density of $\B{x_0}{t_0}$ can be computed exactly, making it useful for validating later algorithms \cite{nye14diffusion}. 
In this case, $\BHV{4}$ consists of three rays joined at the origin, and the density is a superposition of Gaussian terms on each copy of $\R_{\geq 0}$, some with negative weights. 
As shown in \cite{nye14diffusion}, when data are modelled as i.i.d. from $\B{x_0}{t_0}$, the maximum likelihood estimator for $x_0$ does not exhibit the same stickiness property as the Fr\'echet mean, suggesting diffusion means (estimators for $x_0$) may be preferable.

A random walk on $\BHV{N}$ which approximates Brownian motion under a certain limit was also defined in \cite{nye2020random}. 
Starting from $x_0$ it proceeds for $m\times(N-3)$ steps by randomly perturbing one edge length at a time via an innovation with variance $t_0/m$. 
At codimension-1 boundaries, it moves uniformly at random to one of the three adjacent orthants. 
The induced distribution on the end point of the random walk was shown to converge to $\B{x_0}{t_0}$ as $m\rightarrow\infty$ (specifically, weak convergence of probability measures). 
However, it is advantageous to use a multivariate innovation to generate random walks, as described in Section~\ref{sec:GGF} below.

\section{Analogs of Gaussians in tree space}\label{sec:Gaussians}

The Euclidean Gaussian has several equivalent constructions, but in BHV tree space, these yield distinct distributions with no canonical choice. 
While the rest of the paper focuses on Brownian motion kernels as Gaussian analogs, this section introduces two alternatives: (i) the Gaussian kernel distribution and (ii) a distribution obtained by firing geodesics, the latter used later for the bridge proposal mechanism. 

\subsection{Gaussian kernel distribution}\label{sec:KDE}

The Gaussian kernel distribution is the distribution on $\BHV{N}$ with probability density function specified in Equation~\eqref{equ:KDE} with location parameter $x_0\in\BHV{N}$ and dispersion parameter $t_0>0$. 
In Euclidean space, the density is precisely that of an isotropic Gaussian distribution with mean $x_0$ and variance $t_0$. 
The distribution was used in \cite{wey14} to construct kernel density estimates. 
The normalising constant in~\eqref{equ:KDE} depends on $x_0$ and $t_0$.
If $x_0$ lies in the interior of a maximal orthant, most mass lies within that orthant, and the constant approximates that of an $(N-3)$-dimensional Gaussian. 
However, when $x_0$ has small edge lengths (relative to $t_0^{1/2}$), other orthants contribute significantly. 
Each contribution depends on the combinatorics of the geodesics from trees in that orthant to $x_0$, so exact computation is difficult. 
In \cite{weyenberg2016normalizing}, an approximate normalising constant was obtained by computing the integral of the Gaussian distribution over the orthant containing $x_0$, and computing the contribution from all other orthants using cone paths to $x_0$. 
This approximation is poor when $x_0$ contains a mixture of short edges and long edges.
Because of these issues, we propose the Brownian kernel distribution as a more practical alternative for inference.
Although simulation methods for the Gaussian kernel have not been explicitly developed, they can be implemented using a Metropolis-Hastings MCMC algorithm using random walks as proposals.

\subsection{Gaussian via geodesic firing}\label{sec:GGF}

Here we describe the Gaussian via Geodesic Firing, or GGF, distribution with location parameter $x_0 \in \BHV{N}$ and dispersion $t_0 > 0$, initially assuming $x_0$ is fully resolved.
The distribution is defined by the following sampling procedure. 
First, draw $\vec{v}\in\R^{N-3}$ from $N(0,t_0\id_{N-3})$ and fire a geodesic from $x_0$ in direction $\vec{v}$ within the orthant containing $x_0$.
If a codimension-1 boundary is reached, the geodesic continues into one of the two adjacent maximal orthants, chosen uniformly. 
This process continues until the geodesic reaches length $|\vec{v}|$, yielding a sample $x \in \BHV{N}$.

The direction $\vec{v}$ almost surely avoids higher-codimension regions, so $\Gamma_{x_0,x}$ is simple almost surely. 
It follows that GGF has regions of zero density in tree space: for example for certain maximal orthants $\orthant{}\subseteq\BHV{N}$, $\Gamma_{x_0,x}$ will be a cone path for all $x\in\orthant{}$, and so GGF will have zero density there.
The probability density function is:
\begin{equation}\label{equ:GGF-dens}
\GGFdens(x| x_0, t_0) =\begin{cases} \left(\frac{1}{2}\right)^{\nu(x,x_0)}\frac{1}{(2\pi)^\N t_0^{\N/2}}\exp-\frac{1}{2t_0}{\BHVmetric{x}{x_0}}^2 & \text{if $\Gamma_{x,x_0}$ is simple}\\
0 & \text{otherwise}
\end{cases}
\end{equation}
where $\N=N-3$, $\Gamma_{x,x_0}$ is the geodesic from $x$ to $x_0$ and $\nu(x,x_0)$ is the number of codimension-$1$ boundaries traversed by the geodesic.

At codimension-1 boundaries, GGF extends geodesics into the two adjacent maximal orthants with probability $1/2$ each. 
An alternative in \cite{nye2020random} extends with probability $1/3$ into each adjacent orthant and reflects off the boundary with probability $1/3$. 
Though this variant shares the same support as GGF, its associated random walk converges more slowly to Brownian motion and is more computationally complex. 
We therefore do not use it further.

The random walk in \cite{nye2020random} modifies one edge length at a time, which limits its use in bridge construction in $\BHV{N}$. 
A more flexible alternative is defined as follows:

\begin{algorithm}\label{alg:RW-GGF}
Let $x_0\in\BHV{N}$ be fully resolved, and let $y_0=x_0$. 
Then, for $j=1,2,\ldots,m$ sample $y_j$ from $\GGF(y_{j-1},t_0/m)$.
\end{algorithm}

The distribution of $y_m$ conditional on $x_0,t_0$ is denoted $\Wm{x_0}{t_0}$. 
We claim that this random walk also converges to Brownian motion. 

\begin{lemma}\label{lem:conv}
\begin{equation*}
\Wm{x_0}{t_0}\xrightarrow{w}\B{x_0}{t_0}
\end{equation*}
as $m\rightarrow\infty$, where $w$ denotes weak convergence.
\end{lemma}

While \cite{nye2020random} proves convergence if the random walk innovation includes reflection off codimension-1 boundaries, the GGF version does not satisfy this condition. 
However, near a codimension-1 boundary, in the limit that $m\rightarrow\infty$ the boundary will be crossed repeatedly before the random walk moves away. 
As a result, in the limit, the final orthant is uniformly selected from the three adjacent orthants. 
The convergence argument in \cite{nye2020random} therefore carries over to the GGF random walk with this minor modification.  


\section{Sampling bridges}\label{sec:bridging}

Let $\vec{Y}_{[0,m]}=(Y_0,\ldots,Y_m)$ be the random walk from Algorithm~\ref{alg:RW-GGF}, so that given $Y_0=x_0$, the probability density function of $\vec{Y}_{[1,m]}$ is
\begin{equation}\label{equ:RW-like}
f_{\{\vec{Y}_{[1,m]}|Y_0\}}(\vec{y}_{[1,m]} \;|\; x_0, t_0) = \prod_{j=1}^m \GGFdens(y_j \;|\; y_{j-1},t_0)
\end{equation}
where $\vec{y}_{[1,m]}=(y_1,\ldots,y_m)$ and $y_0=x_0$.  
Our inference methods require sampling $\vec{Y}_{[1,m-1]}$ conditioned on $Y_m=x_\star$ and $Y_0=x_0$, where $x_0,x_\star\in\BHV{N}$.  
The conditional probability density function is
\begin{equation}\label{equ:bridgedensity}
f_{\{\vec{Y}_{[1,m-1]}|Y_m,Y_0\}}\left(\vec{y}_{[1,m-1]} \;|\; x_\star,\, x_0,\, t_0\right) = \frac{f_{\{\vec{Y}_{[1,m]}|Y_0\}}\left((y_1,\ldots,y_{m-1},x_\star) \;|\; x_0, t_0\right)}{\walkdens(x_\star\;|\; x_0, t_0; m)} 
\end{equation}
where $\walkdens(\cdot\;|\; x_0, t_0; m)$ is the density function of $\Wm{x_0}{t_0}$.
In this section we describe proposal distributions for sampling such random walk paths, which we refer to as \emph{bridges}.
These proposals approximate the target distribution but, within our MCMC framework, are used to sample from it exactly.

A challenge in constructing bridges is that for every step $y_{j-1},y_{j}$ in a GGF random walk, $\Gamma_{y_{j-1},y_{j}}$ must be simple. 
Thus, a valid bridge will trace a sequence of NNI operations connecting $x_0$ to $x_\star$. 
Finding the shortest such sequence is NP-complete \cite{nni96}, and only when the geodesic $\geod{x_0}{x_\star}$ is simple does it directly provide a valid NNI sequence. 
In the worst case, the geodesic is a cone path and so gives no information about the sequence of NNI's equired to wind around the origin from $x_0$ to $x_\star$. 

We begin by introducing a basic proposal mechanism, to which we then add certain refinements to improve performance. 
The section concludes with an MCMC sampler that uses these proposals to draw from the target distribution \eqref{equ:bridgedensity}.

\subsection{The basic proposal}

We mimic the construction of a Euclidean Brownian bridge, for which exact sampling from the conditional is straightforward. 
Temporarily abusing notation, consider a Gaussian random walk $Y_0,Y_1,Y_2,\ldots,Y_m$ from $x_0$ to $x_\star$ in $\R^{N-3}$, with each step defined by $Y_{j}|y_0,\ldots,y_{j-1}\sim N\left( y_{j-1},\,\frac{t_0}{m}\id_{N-3} \right)$. 
Conditioned on $\vec{Y}_{[0,j-1]}=\vec{y}_{[0,j-1]}$ and $Y_m = x_\star$, the distribution of $Y_j$ is
\begin{equation}\label{equ:Brownianbridge}
N\left( \frac{m-j}{m-j+1}y_{j-1} + \frac{1}{m-j+1}x_\star, \frac{m-j}{m-j+1}\frac{t_0}{m}\id_{N-3} \right).
\end{equation}
The mean is a point a proportion $1/(m-j+1)$ along the line segment from $y_{j-1}$ to $x_\star$ in $\R^{N-3}$. 
In BHV tree space, the analogy is to sample $y_{j}\in\BHV{N}$ by randomly perturbing the point the same proportion along the geodesic $\Gamma_{y_{j-1},x_\star}$. 
This yields Algorithm~\ref{alg:basic-bridge}.

\begin{algorithm}\label{alg:basic-bridge}
Given $y_0 = x_0$ (fully resolved) and $x_\star \in \BHV{N}$, perform the following steps for $j=1,\ldots,m-1$:
\begin{enumerate}
\item Construct $\Gamma_{y_{j-1},x_\star}$ and let $\mu_{j}=\Gamma_{y_{j-1},x_\star}(1/(m-j+1))$.
\item Sample $y_{j}$ from $\GGF(\mu_j, \varjm{j}{m})$ where
\begin{equation}\label{equ:altertime}
\varjm{j}{m} = \frac{m-j}{m-j+1}\frac{t_0}{m}.
\end{equation}
\item If $\Gamma_{y_{j-1},y_{j}}$ is not a simple geodesic then stop the algorithm and reject the proposal.
\end{enumerate}
\end{algorithm}

The algorithm uses GGF for perturbation, which is computationally efficient and approximates Gaussian behavior locally within an orthant.
(Since $\mu_j$ might be unresolved, GGF must be extended to unresolved starting trees -- see Appendix \ref{app:GGFunresolved}.)
The resulting proposal density for $\vec{Y}_{[1,m-1]}$ is
\begin{equation*}
q(\vec{y}_{[1,m-1]} \,\mid\, y_0=x_0,y_m=x_\star,t_0) = \prod_{j=1}^{m-1}\GGFdens\left( y_{j} \;\middle|\; \Gamma_{y_{j-1},x_\star}\left(\frac{1}{m-j+1}\right) ,\,\varjm{j}{m}\right).
\end{equation*}
In the limit of all points lying within a single orthant and far from any boundaries, the proposal matches the true conditional distribution exactly.  
Note that Algorithm~\ref{alg:basic-bridge} is an independence proposal, but below we describe a proposal which resamples a subsection of a given bridge.
Step (3) of the algorithm is required since the target distribution (Equation~$\eqref{equ:bridgedensity}$) is only supported on paths $\vec{y}_{[1,m-1]}$ for which which $\Gamma_{y_{j-1},y_{j}}$ is simple $j=1,\ldots,m$.
Depending on the type of geodesic between $x_0$ and $x_\star$, this condition is often not met, and we address this in Sections~\ref{sec:singularitybudget} and~\ref{sec:partialbridge}.

An additional important issue is that when $y_{j}$ is sampled from GGF$(\mu_{j}, \varjm{j}{m})$ at step (2), the support of the proposal does not necessarily include all possible random walk paths between $x_0$ and $x_\star$. 
This is because the support of $\GGFdens(\cdot | \mu_{j}, \varjm{j}{m})$ might not contain all points in the support of the random walk step $\GGFdens(\cdot | y_{j-1}, t_0/m)$ from $y_{j-1}$. 
As a result, there are potential issues with convergence to the correct stationary distribution in the MCMC scheme for sampling bridges. 
To address this, a mixture distribution is used at step (2):
\begin{equation}\label{equ:mixtureprop}
Y_{j}\;|\;y_{j-1},x_0,x_\star,t_0\:\:\sim w(\mu_{j})\GGF\left(\mu_{j}, \varjm{j}{m}\right)
+[1-w(\mu_{j})]\GGF\left( y_{j-1}, {t_0}/{m} \right)
\end{equation}
where $0< w(\mu_{j})< 1$ is a weight depending on $\mu_{j}=\Gamma_{y_{j-1},x_\star}(1/(m-j+1))$. 
Any weight $w(\mu_{j})<1$ ensures the proposal includes all possible GGF random walks, but the GGF perturbation of $\mu_{j}$ is more affected by regions of zero density when $\mu_j$ is close to singularities with codimension $\geq 2$. 
To manage this, we choose $w(\mu_j)$ based on the distance from the origin (star tree), since it can be computed quickly:
\begin{equation}\label{equ:weightfunction}
w(\mu_j)= \max\left\{F_{\chi}\left(\frac{{\BHVmetric{\mu_j}{0}}^2}{\varjm{j}{m}},N-3\right),10^{-3}\right\},
\end{equation}
where $F_{\chi}(\cdot,N-3)$ is the CDF of the $\chi^2$ distribution with $N-3$ degrees of freedom.

\subsection{Dealing with singularities}\label{sec:singularitybudget}

\ifx\du\undefined\newlength{\du}\fi\setlength{\du}{22\unitlength}
\begin{figure}
    \centering
\includegraphics[width=0.6\textwidth]{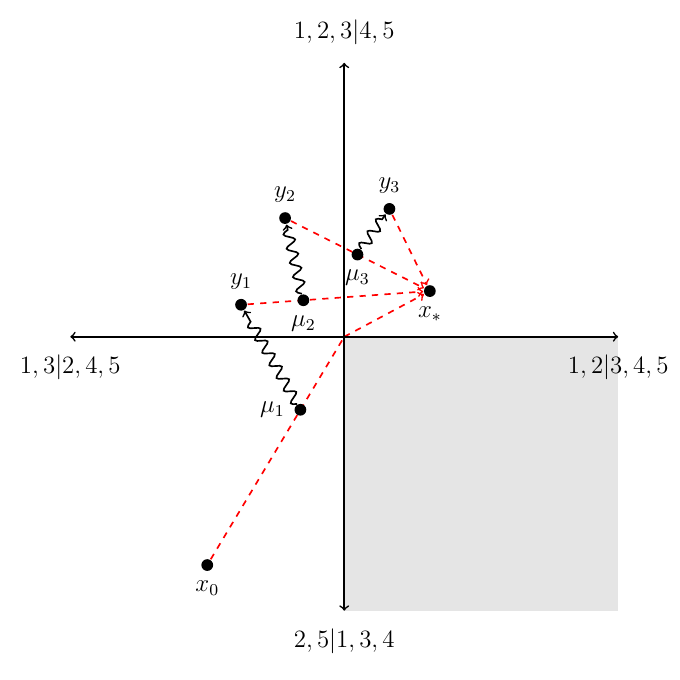}
\caption{
Bridge construction with $m=4$ steps in $\BHV{5}$ using Algorithm~\ref{alg:singularitybridge}. 
Three maximal orthants are shown; the shaded region is excluded from the space. 
Here we use a penalty function $f_p$ that returns $2$ if $\Gamma_{y_{j-1},x_\star}$ contains the origin, and zero otherwise. 
Dashed lines indicate geodesics; waves show GGF perturbations. 
Since $\Gamma_{x_0,x_\star}$ passes through the origin, $p_1=2$ and $\mu_1$ is placed at $\Gamma_{x_0,x_\star}(1/2)$. 
The remaining steps then wind around the origin to $x_\star$; each remaining step has $p_1=p_2=0$ since $y_1$ and $y_2$ are connected via simple geodesics to $x_\star$.
}
    \label{fig:BridgeWithPenalty}
\end{figure}

By design, Algorithm~\ref{alg:basic-bridge} mimics a Euclidean Brownian bridge and works well when $x_0$ and $x_\star$ lie in the interior of the same orthant, yielding high acceptance rates. 
However, departures from this regime result in a lower acceptance probability, often due to the condition in step (3) failing.
Consider the following example.
Suppose that $\Gamma_{x_0,x_\star}$ is a cone path; that $x_0$ is far from any boundary with codimension $\geq 1$; and that  $x_\star$ lies close to the origin. 
Then, if $t_0$ is small, Algorithm~\ref{alg:basic-bridge} will produce a sequence of points $y_1,y_2,\ldots$ lying close to the line segment connecting $x_0$ to the origin. 
The proposal is likely to fail at step (3) in later steps, as too few iterations will remain to wind around the origin to $x_\star$ using simple geodesics.
In order to remedy this and increase the probability of proposing a valid random walk path, at each iteration we allocate a certain budget of steps to wind around singularities on the geodesic $\Gamma_{y_{j-1},x_\star}$. 
This yields the following modified proposal, where the proposal density is given in Appendix~\ref{sec:bridgeMCMCdetails}. 

\begin{algorithm}\label{alg:singularitybridge}
Given $y_0 = x_0$ (fully resolved) and $x_\star \in \BHV{N}$, perform the following steps for $j=1,\ldots,m-1$:
\begin{enumerate}
    \item Construct $\Gamma_{y_{j-1},x_\star}$. Set $ p_j = f_p(\Gamma_{y_{j-1},x_\star})$, where $f_p$ is an integer-valued function called the penalty function.
    \item Let $\gamma$ be the geodesic segment $\Gamma_{y_{j-1},x_\star} [0,1/(m-j+1-p_j)]$.  Define $\mu_{j}$ as follows.
        \begin{enumerate} 
        \item Set $\mu_j=\Gamma_{y_{j-1},x_\star} (1/(m-j+1-p_j))$ if there is no boundary with codimension $> 1$ in $\gamma$; 
        \item otherwise set $\mu_{j}$ to be point on $\gamma$ with codimension $\geq 2$ closest to $y_{j-1}$. 
\end{enumerate}
    \item Sample $y_{j}$ from the mixture distribution in Equation~\ref{equ:mixtureprop}.
    \item If $\Gamma_{y_{j-1},y_{j}}$ is not simple, reject the proposal.
\end{enumerate}
\end{algorithm}

The penalty $p_j$ is the number of iterations budgeted to bypass singularities on $\Gamma_{y_{j-1},x_\star}$, with $f_p$ assigning higher penalties for higher codimension singularities. 
As $p_j$ increases, the size of the step from $y_{j-1}$ towards $x_\star$ increases. 
With the exception of step 2(b), Algorithm~\ref{alg:basic-bridge} is recovered by setting $p_j=0$ for all $j$. 
Step 2(b) is introduced to increase the probability that $\Gamma_{y_{j-1},y_{j}}$ is simple. 
Algorithm~\ref{alg:singularitybridge} is illustrated in Figure~\ref{fig:BridgeWithPenalty}.

In practice, we used the following penalty function, selected via tuning on simulated data.
For a geodesic $\Gamma$, let $\beta(\Gamma)$ be the number of orthants of codimension $>1$ traversed by $\Gamma$ and let $\kappa_i(\Gamma)$ be the codimension of the $i$th such orthant. 
Then  
\begin{equation*}
f_{p}(\Gamma) = \sum_{i=1}^{\beta(\Gamma)}\kappa_i(\Gamma).
\end{equation*}
If $f_{p}(\Gamma_{y_{j-1},x_\star})\geq (m-j-1)$ (so that the penalty exceeds the number of remaining steps), we set $p_j=0$.

\subsection{Partial bridges and MCMC sampling of bridges}\label{sec:partialbridge}

Even with the modifications in Algorithm~\ref{alg:singularitybridge}, the proposal can have a very low acceptance probability. 
We therefore use proposals for a random walk path $\vec{Y}_{[0,m]}^*$ conditional on an existing valid path $\vec{Y}_{[0,m]}=\vec{y}_{[0,m]}$ from $x_0$ to $x_\star$, which update only a segment of the bridge path, replacing steps $a+1,\ldots,a+l$ of $\vec{Y}_{[0,m]}$ rather than resampling the entire bridge. 
Formally, there is a different proposal for each fixed value of $a$ and $l$, each of which is in detailed balance with the target distribution. 
Given $a$ and $l$, new points $Y^*_{a+1},\ldots,Y^*_{a+l}$ are proposed using  Algorithm~\ref{alg:singularitybridge} but with $x_0$ replaced by $y_{a}$; $x_\star$ replaced by $y_{a+l+1}$; and $m$ replaced by $l+1$.

\begin{algorithm}\label{alg:bridgesampler}
We use MCMC to sample random walk paths $\vec{Y}_{[0,m]}$ from the conditional distribution in Equation~$\eqref{equ:bridgedensity}$, using the partial bridge proposal for $\vec{Y}_{[0,m]}^*$ at step $j$, given a valid path $\vec{Y}_{[0,m]}$ at step $j-1$. 
See Appendix~\ref{sec:bridgeMCMCdetails} for details, including the proposal density and a proof that the induced Markov chain has the intended stationary distribution.  
\end{algorithm}

\section{Bayesian inference of source and dispersion parameters}\label{sec:inference}

\subsection{Target distribution}

Suppose $\vec{x}=(x_1,\ldots,x_n)$ is a random sample of trees $x_i\in\BHV{N}$. 
We model each $x_i$ as an independent draw from $\Wm{x_0}{t_0}$, where $m$ is chosen large enough for the random walk to approximate the Brownian motion $\B{x_0}{t_0}$. 
We aim to infer the source and dispersion parameters $x_0 \in \BHV{N}$ and $t_0 > 0$.
Since the density of $\Wm{x_0}{t_0}$ is intractable, we augment the model with the latent random walk paths $\vec{Y}_{i,[1,m-1]}=(Y_{i,1},\ldots,Y_{i,m-1})$ between $x_0$ and $x_i$, $i=1,\ldots,n$. 
For simplicity of notation, for $i=1,\ldots,n$ we let $\vec{Y}_i$ denote the vector of random variables $\vec{Y}_{i,[1,m-1]}$, and let $\vec{y}_i=\vec{y}_{i,[1,m-1]}=(y_{i,1},\ldots,y_{i,m-1})$ be their realisations in $\BHV{N}$. 
Given a prior $\pi(x_0,t_0)$, the joint density of the parameters, random walk paths and data is
\begin{align*}
f(\vec{y}_1,\ldots,\vec{y}_n,\vec{x},x_0,t_0; m) &= \pi(x_0,t_0)\prod_{i=1}^n
\GGFdens(x_i\;|\;y_{i,m-1},t_0/m)\GGFdens(y_{i,1}\;|\;x_0,t_0/m)\\ &\qquad \times\prod_{j=2}^{m-1}\GGFdens(y_{i,j}\;|\;y_{i,j-1},t_0/m)
\\
&=
\pi(x_0,t_0)\prod_{i=1}^n \prod_{j=1}^{m}\GGFdens(y_{i,j}\;|\;y_{i,j-1},t_0/m)
\end{align*}
where the second equality follows by fixing the convention $y_{i,0}=x_0$ and $y_{i,m}=x_i$, $i=1,\ldots,n$. 
The posterior distribution is given by Bayes' theorem
\begin{equation}\label{equ:maintarget}
f(\vec{y}_1,\ldots,\vec{y}_n,x_0,t_0 | \vec{x}; m)=
\frac{\pi(x_0,t_0)\prod_{i=1}^n \prod_{j=1}^{m}\GGFdens(y_{i,j}\;|\;y_{i,j-1},t_0/m)}{f(\vec{x})}
\end{equation}
where the denominator is the marginal probability density function of the data, integrating over all latent paths and parameters. 
A Metropolis-within-Gibbs MCMC algorithm is used to sample from this distribution, avoiding the need to compute the intractable denominator. 
Marginalising over the latent paths yields posterior samples for $(x_0, t_0)$.

\subsection{Proposals}

At each iteration of the MCMC algorithm we sweep through the following proposals: (i) a proposal for $\vec{Y}_i$, $i=1,\ldots,n$; (ii) a proposal for $x_0$; and (iii) a proposal for $t_0$. 
Each proposal is specified below, and the acceptance probabilities are given in Appendix~\ref{sec:mainMCMCdetails}. 
Comments about the performance of the proposals are given in Section~\ref{sec:simulation}. 

\textbf{Proposal for the random walk path $\vec{Y}_i$.} In each iteration of the MCMC algorithm, a starting position $a$ and length $l$ for partial bridges are randomly sampled as in Algorithm~\ref{alg:bridgesampler}, and for $i=1,\dots,n$, the partial bridge proposal with parameters $(a,l)$ is used to propose a bridge $\vec{Y}_i^*$ given the current bridge $\vec{Y}_i=\vec{y}_i$. 

\textbf{Proposal for $x_0$.} This proposal works as follows: we first sample $x_0^*$ from $\GGF(x_0,\lambda_0^2)$, for some fixed parameter $\lambda_0>0$; then, conditional on $x_0^*$, we generate a new random walk path $\vec{Y}_i^*$ from $x_0^*$ to $x_i$, $i=1,\ldots,n$, by replacing the first $l$ steps of each bridge $\vec{Y}_i$ via the partial bridge proposal. 
Formally, there is a different proposal for each value of $l$. 
At each iteration in the MCMC sampler, a proposal is chosen by sampling $l$ from a truncated geometric distribution on $0,\ldots,m-1$ with parameter $\alpha_0\in(0,1)$, and the corresponding proposal is employed for all bridges $i=1,\ldots,n$.  
If $l>0$ we replace points $y_{i,1},\ldots,y_{i,l}$ on each bridge $\vec{Y}_i=\vec{y}_i$ using the partial bridge proposal from Section~\ref{sec:partialbridge} with start point $a=0$ and length $l$, and with $x_0$ replaced by $x_0^*$.
If $l=0$ then only a new value of $x_0$ is proposed, and the bridge steps are unchanged for all $i$. 

\textbf{Proposal for $t_0$.} A simple log-normal proposal is used to update $t_0$: 
\begin{equation*}
\left(t_0^*\;|\;t_0\right)\sim \exp\left( N(\log(t_0),\sigma_0^2) \right)   
\end{equation*}
for a fixed parameter $\sigma_0>0$. 
The bridges are retained. 

\subsection{Priors} 

We assume $x_0$ and $t_0$ to be independent under the prior.
The prior on $x_0$ is uniform over tree topology, with edge lengths chosen to reflect typical divergence levels in phylogenies.
In most biological data sets, path lengths between leaves are less than $1$, since a value of $1$ represents high sequence divergence. 
Letting $\hat{d}_\text{BHV}$ denote the metric between trees which is the product metric of $d_{\text{BHV}}$ and the Euclidean metric on the vector of pendant edge lengths, an upper bound for $\hat{d}_\text{BHV}(0,x)^2$ for realistic phylogenies occurs when $x$ is the star tree with pendant edges of length $1/2$, and we denote this as $D^2_N=N/4$. 
Since $\hat{d}_\text{BHV}\geq d_\text{BHV}$, we use this as an approximate upper bound on ${d}_\text{BHV}(0,x)^2$. 
We assume that $d(0,x_0)$ has a half-normal distribution so that
\begin{equation}\label{equ:x0Prior}
d(0,x_0)^2 \sim\text{Ga}(1/2, 3.3175/D_N^2)
\end{equation}
and the $99\%$ quantile of the distribution is $D_N^2$. 
Alternatively, any tree-valued prior used in phylogenetic inference could be used for $x_0$, such as the commonly used exponential prior on edge lengths.
However, this prior puts very little mass near the origin, and so is less suitable. 

Displacement through BHV tree space by a squared distance $D^2_N$ would correspond to loss of evolutionary signal from the source tree.
For a Brownian motion in $\R^{N-3}$, the expected squared distance from the source is $(N-3)t_0$. 
Using this as an approximation to BHV tree space, our prior on $t_0$ is
\begin{equation*}
t_0 \sim \text{Exp}(4.61(N-3)/D_N^2),
\end{equation*}
so that the $99\%$ quantile of $(N-3) t_0$ is $D_N^2$.

\subsection{Other details}\label{sec:otherdetails}

Initial values for $x_0$ and $t_0$ are based on an approximate Fréchet mean of the sample $x_1, \ldots, x_n$, calculated via a few iterations of an algorithm in~\cite{sturm2003}. 
We set $x_0$ to the data point closest to this mean and $t_0$ to the estimated Fréchet sample variance (the value of the sum in Equation~\eqref{equ:frechet}).

The choice for the number of steps $m$ used for the random walk is guided by a balance between computational speed and the desired level of approximation of the random walk to the Brownian motion. 
Since $\B{x_0}{t_0}$ is supported on all of $\BHV{N}$, $m$ must be large enough for the walk to reach any orthant. 
Known bounds on the minimum number of NNI operations between trees \cite{nni96} guide how $m$ should scale with $N$. 
Forward simulation of random walks for different values of $m$ for fixed $(x_0,t_0)$ can also be used to assess convergence on summary statistics of such samples (e.g. the number of topologies displayed).

\subsection{Bayesian consistency}

We conclude this Section by proving that the posterior distribution is consistent,
in the sense that in the limit of observing an infinite number of data points from the Brownian motion kernel, the posterior distribution for $x_0$ concentrates around the true source tree. 
We prove consistency under the assumption that $t_0$ is fixed and known, and that the data are independently and identically distributed according to $B(x_0,t_0)$. 
Specifically, we apply Doob’s theorem~\cite{DoobLJ1949ATM} in the general formulation of~\cite{MillerJeffreyW2018Adto}, which covers Borel subsets of Polish spaces for both data and parameter spaces.

\begin{theorem}\label{thm:consistency}
Let $\BHV{N}^{(0)}\subset\BHV{N}$ denote the union of the interior of all maximal orthants, and let $\mathcal{A}$ be the Borel $\sigma$-algebra on $\BHV{N}^{(0)}$.  
Let $K(x,\epsilon)$ denote the open ball in $\BHV{N}$ centred at $x$ with radius $\epsilon$. 
Suppose $X_0$ is distributed according to the prior $\pi$ for $x_0$, and that $X_i|X_0\sim B(X_0,t_0)$ independently, $i=1,2,\ldots$. 
Then there exists a set $A\in\mathcal{A}$ with $\pi(A)=1$, such that when $x_0\in A$ and $x_1,x_2,\ldots$ are a sequence of independent observations from $B(x_0,t_0)$, 
$$
\pr{X_0\in K(x_0,\epsilon)\;|\;(X_1=x_1,\ldots X_n=x_n)}\to 1
$$
for any $\epsilon>0$ as $n \to \infty$ almost surely over the measure $B(x_0,t_0)$ on the observations. 

\end{theorem}

We give the proof of Theorem~\ref{thm:consistency} in  Section~\ref{sec:consistencyproof} of the Appendix, along with further comments.

\section{Marginal likelihood}\label{sec:marginalLikelihood}

One advantage of our probabilistic approach over existing least squares methods is that hypothesis tests can be performed very naturally via existing Bayesian methodology.   In particular, we want to compare different source trees for a given data set. 
In this section we consider $t_0$ to be fixed and known. Suppose we have data $\vec{x}=(x_1,\ldots,x_n)$, which we model as independent draws from a Brownian motion kernel $B(x_0,t_0)$, approximated by the random walk kernel $\Wm{x_0}{t_0}$ with some fixed $m$. We are interested in testing hypotheses of the form
$$
H_0: x_0=x_0'\text{ and }H_1:x_0=x_0'',
$$
by using Bayes factors.
For this we need to estimate the log marginal likelihood
\begin{align*}
    \log f(\vec{x}|x_0,t_0) =& \sum_{i=1}^n\log\walkdens(x_i\;|\; x_0,t_0; m)\\  =&\sum_{i=1}^n\log\int \prod_{j=1}^{m}\GGFdens(y_{i,j}\;|\;y_{i,j-1},t_0/m)d(y_{i,1},\ldots y_{i,m-1})
\end{align*}
where $\walkdens(\cdot\;|\; x_0, t_0; m)$ is the probability density function of $\Wm{x_0}{t_0}$, $y_{i,0}=x_0$ and $y_{i,m}=x_i$, $i=1,\ldots,n$.
We will adopt the notation $d\vec{y}_i=d(y_{i,1},\ldots y_{i,m-1})$.
The independence of the bridges conditioned on fixed values of $x_0$ and $t_0$ means that we can calculate estimates for the value of $\walkdens(x_i\;|\; x_0,t_0; m)$ separately for each $i=1,\dots,n$. 
Since the integral $\int \prod_{j=1}^{m}\GGFdens(y_{i,j}\;|\;y_{i,j-1},t_0/m)d\vec{y}_{i}$ decomposes as 
$$
\int \GGFdens(x_i\;|\;y_{i,m-1},t_0/m) \prod_{j=1}^{m-1}\GGFdens(y_{i,j}\;|\;y_{i,j-1},t_0/m)d\vec{y}_{i},
$$
it follows that the integral can be seen as a Bayesian marginal likelihood calculation where the prior density is an unconditioned $(m-1)$-step random walk density with dispersion $t_0(m-1)/m$ and the likelihood is the density of the last step of the $m$-step random walk to the fixed end point $x_i$. 

We estimate marginal likelihoods using three existing methods from Euclidean settings, allowing both accuracy checks and identification of the most effective approach in BHV tree space:
(i) Chib's method~\cite{ChibSiddhartha2001MLFt}; 
(ii) the `bridge' method from~\cite{MengXiao-Li1996SRON}; and
(iii) the method of (generalised) stepping stone sampling from~\cite{FanYu2011Capm}. 
Since the term `bridge' is already used in this paper to mean something different, we will instead use the term `tunnel sampling' to refer to method (ii). 
Chib's method and the tunnel sampling method had similar performance in terms of variance of the estimators and computation speed, and out-performed the stepping-stone method. 
As a result we will focus on Chib's method here, with details of the other methods in the Appendix. 
  
Chib's method, which is the single block method from~\cite{ChibSiddhartha2001MLFt}, has the following form in the simplest case of a single data point $x_\star$. 
First, for some fixed bridge $\tilde{\vec{y}}_{[1,m-1]}$ from $x_0$ to $x_\star$, an estimate of the left hand side of Equation~$\eqref{equ:bridgedensity}$ is calculated; then by rearranging Equation~$\eqref{equ:bridgedensity}$ an estimate of $\walkdens(x_\star\;|\; x_0, t_0; m)$ is obtained since the numerator is easy to calculate. 
The estimate of the left hand side of Equation~$\eqref{equ:bridgedensity}$ is obtained by combining samples from the conditional distribution of the bridge $\vec{y}_{[1,m-1]}$ between $x_0$ and $x_\star$, obtained via Algorithm~\ref{alg:bridgesampler}, and samples from the independence proposal distribution (Algorithm~\ref{alg:singularitybridge}). 
In fact, instead of using a single bridge $\tilde{\vec{y}}_{[1,m-1]}$ for this calculation, it is convenient to reduce the variance of the Chib estimator by repeating the calculation for a number of bridges drawn from the conditional distribution. 
This has the advantage of stabilising the estimate for the marginal likelihood, without increasing the number of samples from the conditional distribution that need to be simulated. 
The only increase in computational cost is that needed to calculate the estimate for each of the selected bridges. 
A detailed algorithm is given in the Appendix (Algorithm~\ref{alg:ML:Chib}).

\section{Simulation study}\label{sec:simulation}

In this section we validate the performance of our algorithms in three different ways. (i) We simulate a number of bridges between fixed points $x_0$ and $x_\star$ in $\BHV{10}$ using Algorithm~\ref{alg:bridgesampler}, and show that the MCMC mixes over paths that pass through different sets of topologies.
(ii) Marginal likelihoods are estimated in the scenario of fixed $x_0,t_0$ and a single data point $x_\star$. 
On $\BHV{4}$ the marginal likelihood can be calculated analytically and this is used to validate our algorithms. 
(iii) We perform the inference described in Section~\ref{sec:inference} on data sets in $\BHV{10}$ simulated via random walk. We show how the marginal posterior samples for $x_0$ and $t_0$ are concentrated near the true values used to simulate the data set. 
Finally, we discuss the numerous challenges presented by this kind of inference. 
These simulations form a small snapshot of much more extensive testing performed to validate the algorithms. 

\subsection{Bridge simulations}\label{sec:BridgeSims}

To assess the performance of Algorithm~\ref{alg:bridgesampler}, we randomly generated three pairs of trees $x_{0,i},x_{\star,i}$, $i=1,2,3$ with $N=10$ taxa, and simulated bridges between each pair. 
The pairs of trees were selected to yield geodesics $\Gamma_{x_{0,i}, x_{\star,i}}$ traversing increasingly complex (higher-codimension) regions; a description of each geodesic is given in Table~\ref{tab:bridgeSimsAcceptanceRates}.

We sampled bridges using Algorithm~\ref{alg:bridgesampler} with $m = 50$ steps for each pair. 
After a burn-in of $10^4$ iterations, the sampler ran for $4 \times 10^5$ iterations, thinning every 100th bridge. 
We used $\alpha_b = 0.01$  to select the proposed partial bridge length (tuned for the cone path case, $i=3$) and applied the same setting across all three cases to compare acceptance rates. 
On average, 23 steps were proposed for update at each iteration.

\begin{center}
\begin{table}[ht]
\centering
\begin{tabular}{llc}
  \hline
 Pair $i$ & Geodesic description & Acceptance rate\\ 
  \hline
 1  & two codimension-$2$ boundaries & 40.8\%\\ 
 2  & one codimension-$5$ boundary & 27.5\%\\
 3  & cone path & 18.1\% \\ 
   \hline
\\
\end{tabular}
\begin{caption}{\label{tab:bridgeSimsAcceptanceRates}Acceptance rates for the partial bridge proposal when simulating bridges between two fixed endpoints. 
}
\end{caption}
\end{table}
\end{center}

The acceptance rates for the partial bridge proposal are shown in Table~\ref{tab:bridgeSimsAcceptanceRates}. 
(The acceptance rate is really the average over different values of $a$ and $l$.)
The acceptance rate is highest for the least complex geodesic and lowest for the cone path geodesic. 
Figure~\ref{fig:bridgeSimsCountsPlot} displays the number of distinct topologies at each step, showing that the third case explores the most topologies.
Overall, a large number of topologies are visited, and together with traceplots of the log-likelihood (Figure~\ref{fig:bridgeSimsLikelihoodPlots}), this suggests the chains mix well. 

\begin{figure}
\begin{center}
\includegraphics[width=0.65\textwidth]{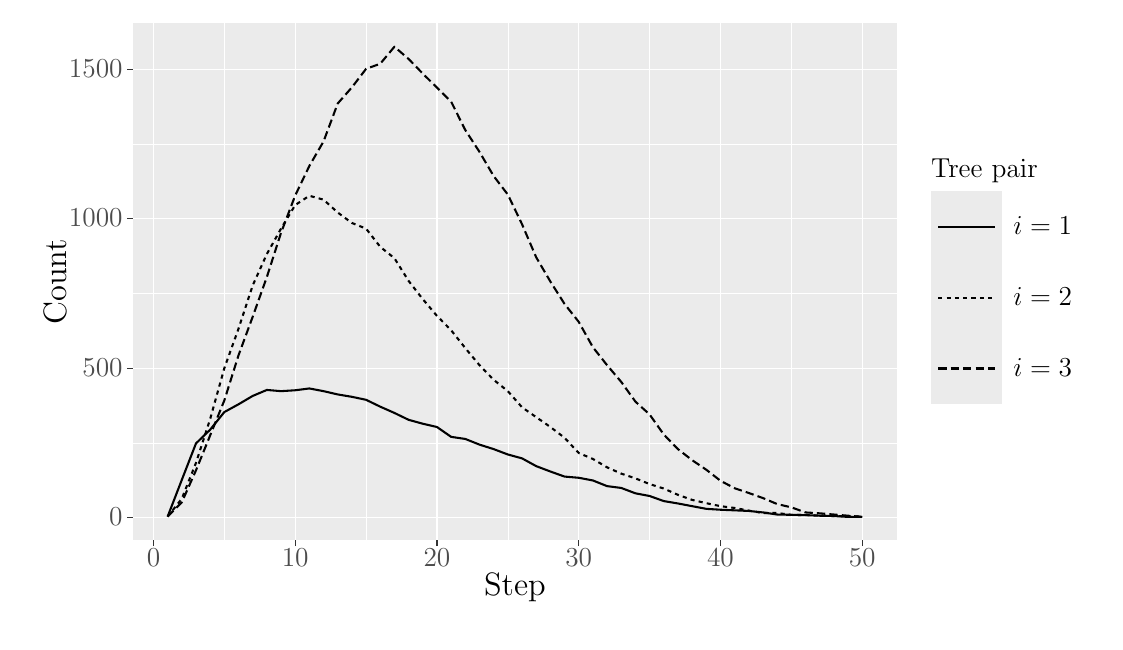}
\begin{caption}{\label{fig:bridgeSimsCountsPlot} 
The number of distinct topologies displayed at each step in samples of bridges simulated between three different sets of fixed endpoints in $\BHV{10}$ ($4\times 10^3$ bridges in each sample).
}
\end{caption}
\end{center}
\end{figure}

\subsection{Marginal likelihood for single data points}

For $N=4$ and a single data point $x_\star$, the marginal likelihood can be computed exactly using the closed-form expression for $\B{x_0}{t_0}$ from~\cite{nye14diffusion}. As an initial test, we fixed $x_0$ and estimated the marginal likelihood 100 times for different $x_\star$ values, assuming $x_\star$ was drawn from $\Wm{x_0}{t_0}$  with $m = 20$ and $t_0 = 0.25$. 
Median estimates from the Chib method and the true values are shown in Figure~\ref{fig:3SpiderResultsChib}. 
The variance of the estimates was negligible, and the other marginal likelihood estimators produced nearly identical results. 
Marginal likelihoods were estimated for the random walk model and hence there is a small discrepancy between the estimated values and the true value for the Brownian motion kernel.

\begin{figure}
\begin{center}
\includegraphics[width=0.4\textwidth]{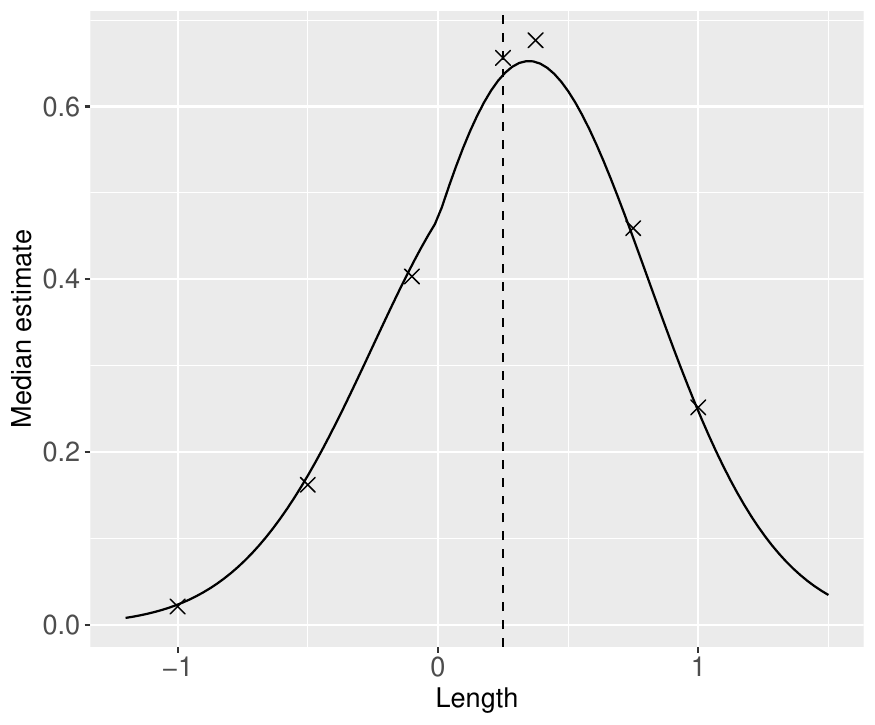}
\begin{caption}{\label{fig:3SpiderResultsChib}
The results of estimating marginal likelihoods $\BHV{4}$. 
The curve represents the true value of the marginal likelihood using the exact Brownian motion kernel for $N=4$ in~\cite{nye14diffusion}. The dashed vertical line shows the position of $x_0$. 
The positive axis is the orthant containing $x_0$; the negative axis represents the other two orthants in $\BHV{4}$. 
The crosses show the median of $100$ estimates of the marginal likelihood using the Chib estimator.
}
\end{caption}
\end{center}
\end{figure}
 
For $N>4$ taxa, we cannot calculate the marginal likelihood exactly.
The different estimation algorithms were run $100$ times on the three pairs of trees used in Section~\ref{sec:BridgeSims}, using $x_{0,i}$ as the source and $x_{\star,i}$ as the data point, for cases $i=1,2,3$. 
Results are shown in the Appendix, but the estimated values were consistent between the three methods, with the Chib and tunnel methods performing best.

\subsection{Inference of $(x_0,t_0)$}\label{sec:simstudyInference}

We tested the MCMC sampler from Section~\ref{sec:inference} on simulated data to infer $(x_0, t_0)$. 
A source tree $x_0$ with $N = 10$ taxa was generated with edge lengths from a Gamma$(2, 0.5)$ distribution (Figure~\ref{fig:InferenceSourceTree}). 
The edge separating taxa 1 and 10 is particularly short, placing $x_0$ near a codimension-1 face $E_0$. 
Data sets of size $n=50$ were generated by forward simulating random walks with $m=2\times 10 ^3$ steps from $x_0$. 
Representative values of the dispersion $t_0$ were selected by counting the number of distinct topologies in larger samples from $\Wm{x_0}{t_0}$ (see Appendix Figure~\ref{fig:distinctTopsPlot}).  
Based on the plot, we choose $t_0=\{0.01,0.1,0.3,0.5\}$, so that the data sets contained respectively $35,49,49,48$ different topologies on $n=50$ trees. 

\begin{figure}
\begin{center}
\includegraphics[width=0.4\textwidth,]{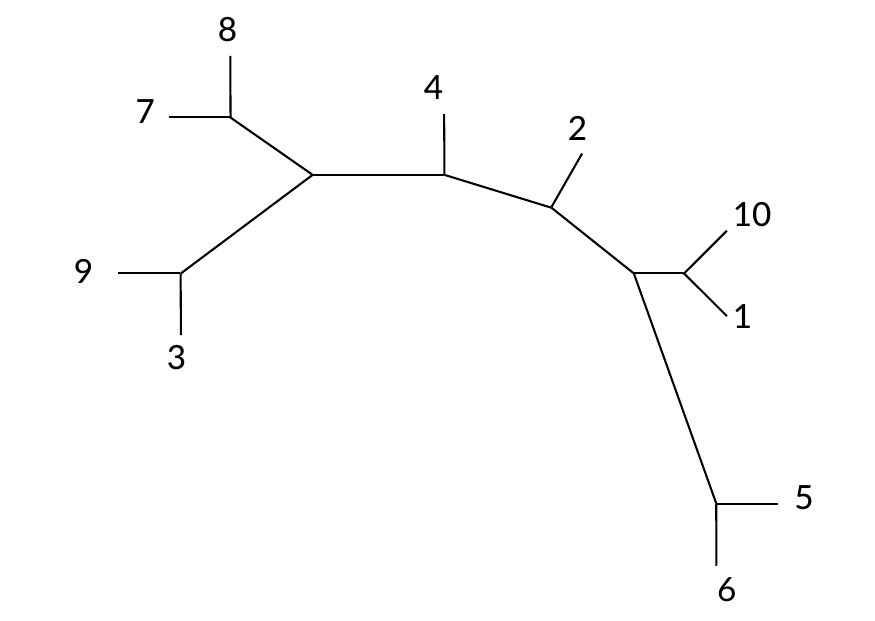}
\begin{caption}{\label{fig:InferenceSourceTree} The source tree $x_0$ used to simulate data setsin Section~\ref{sec:simstudyInference}. 
Internal edge lengths are drawn to scale. Pendant edge lengths are arbitrary. Note the short edge leading to the cherry $(1,10)$. 
}
\end{caption}
\end{center}
\end{figure}

We ran the inference scheme for $(x_0,t_0)$ on each data set. 
Burn-in and thinning details are in Appendix Table~\ref{tab:inferenceTimings}. 
The $t_0 = 0.5$ dataset required a longer burn-in due to increased topological complexity, which lowers bridge acceptance rates. 
To assess whether the chains had run sufficiently long, we tracked the cumulative proportions of visited topologies in the posterior sample for $x_0$; plots in Appendix Figure~\ref{fig:inferenceTopPropPlots} suggest representative sampling, although this was slower for $t_0 = 0.3$ and $0.5$. 
Proposal parameters and acceptance rates are in Appendix Table~\ref{tab:InferenceParamsAndAcceptanceRates}. 

\begin{figure}
\begin{center}
\includegraphics[width=0.8\textwidth]{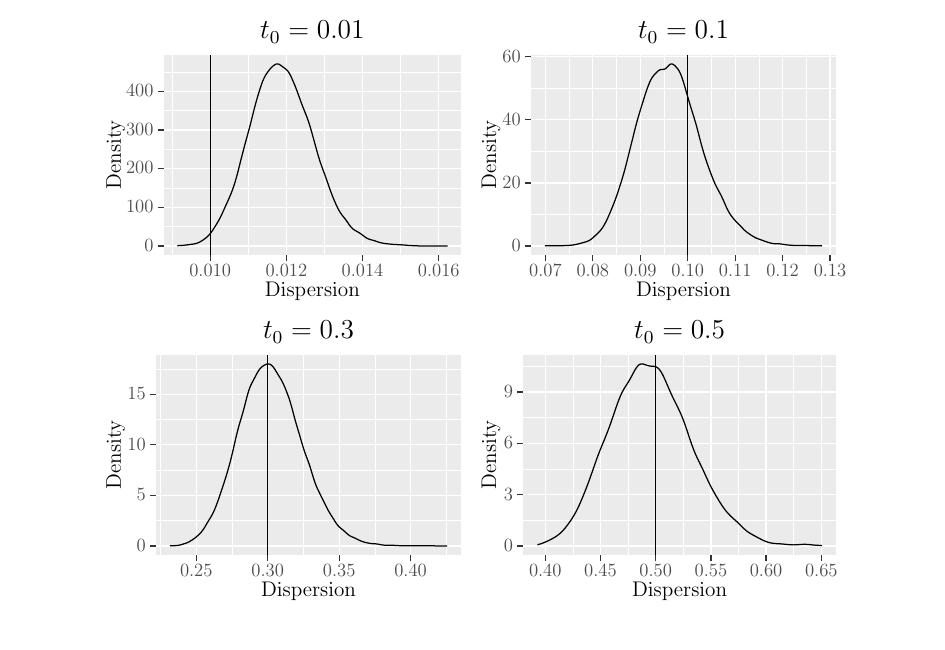}
\begin{caption}{\label{fig:inferenceDispKDEs} 
Kernel density estimates of the marginal posterior density of $t_0$ for inference performed on simulated data sets. 
The vertical line on each plot shows the true value of $t_0$ used to simulate the data set. 
On this scale the prior was close to zero. 
}
\end{caption}
\end{center}
\end{figure}

Posterior distributions of $t_0$ (Figure~\ref{fig:inferenceDispKDEs}) were concentrated near the true values. 
Posterior probabilities for the topology of $x_0$ are given in Table~\ref{tab:posteriorPercentages}. 
In all cases, the posterior for $x_0$ is concentrated in the three maximal orthants adjacent to $E_0$ (top three rows of Table~\ref{tab:posteriorPercentages}); uncertainty arises primarily on account of the short edge in the true tree. 
For $t_0 = 0.01$, $0.1$, and $0.3$, the true topology is the posterior mode, and for $t_0=0.5$ they are separated by a single NNI. 
In contrast, the Fr\'echet mean was the star tree for $t_0=0.3$ and $t_0=0.5$.

Further diagnostics, traceplots, KDEs for edge lengths, and run-time details are provided in the Appendix.

\begin{center}
\begin{table}[ht]
\centering
\begin{tabular}{l|llll}
\hline 
 & \multicolumn{4}{c}{True $t_0$} \\
  \hline
Topology & 0.01 & 0.1 & 0.3 & 0.5 \\ 
  \hline
\textbf{(2,((1,10),(5,6)),(4,((3,9),(7,8))));} & 98\% & 76.7\% & 82.7\% & 9\% \\ 
 (2,(1,(10,(5,6))),(4,((3,9),(7,8))));  & 1.5\% & 3.3\% & 14.4\% & 1.7\% \\ 
  (2,(10,(1,(5,6))),(4,((3,9),(7,8))));  & 0.5\% & 20\% & 0.1\% & 76.7\% \\ 
 (2,(10,(1,(5,6))),((3,9),(4,(7,8))));  & 0.0\% & 0.0\% & 0.0\% & 7.8\% \\  
\end{tabular}
\begin{caption}{\label{tab:posteriorPercentages}Topologies in the marginal posterior sample of $x_0$. 
The column for each simulated data set shows the posterior probability for the 4 topologies listed. Top row: true topology (in bold).  
}
\end{caption}
\end{table}
\end{center}

\subsection{Scalability}

There are a number of difficulties performing MCMC in this setting. 
First, mixing is worse as the size $n$ of the data set increases. 
When updating the $x_0$ parameter, a number $l$ of steps on each bridge is also updated. 
This means $l\times n$ bridge steps are updated in the proposal, rather than $l$ when performing a partial bridge proposal for a single data point, and so the acceptance probability for the $x_0$ move reduces as $n$ increases. 
We therefore change fewer bridge steps in the $x_0$ proposal than the partial bridge proposals (i.e.~$\alpha_0>\alpha_b$).
Secondly, increasing the number of taxa $N$ introduces further difficulties. 
Larger $N$ requires more random walk steps $m$ to approximate Brownian motion, increasing geodesic computations and computational cost. 
Moreover, with larger $N$, more geodesics to $x_0$ pass through high-codimension regions which are harder for the bridge proposal to traverse, reducing proposal acceptance.
Inference remains feasible for larger $N$ and $n$ when the data are more tightly clustered (lower dispersion), which simplifies geodesic structure and improves acceptance rates.
Further comments are made in Section~\ref{sec:discussion}.

\section{Biological example}\label{sec:biologicalex}

We applied our inference methods to a well-known yeast gene tree data set~\cite{rok03}, consisting of $n=106$ gene trees for $N=8$ species (seven from the \textit{Saccharomyces} genus, and one from the \textit{Candida} outgroup). 
The data displayed 26 unique internal splits and 23 topologies. 
The fully-resolved modal topology was displayed by $41$ trees, and it was the same as the majority consensus topology. (The majority consensus toplogy is obtained by taking the union of all splits present in $>50\%$ trees in a sample.)
We also estimated the Fr\'echet mean using $10^5$ iterations of the algorithm from~\cite{sturm2003}, and it had the majority consensus topology. 

Inference for $(x_0, t_0)$ used $m=50$ bridge steps, with $5 \times 10^5$ burn-in iterations and a posterior sample of $5 \times 10^4$ drawn from $5 \times 10^6$ thinned iterations. 
Convergence diagnostics and full MCMC details appear in the Appendix (Figures~\ref{fig:yeastLikelihoodPlot}, \ref{fig:yeastDispTraceplot}, \ref{fig:yeastDispKDEplot},  ~\ref{fig:yeastCumulativePropPlot} and Table~\ref{tab:yeastAcceptanceRates}).
The posterior modal tree and the Fr\'echet mean had the same topology, but the modal tree had longer internal edge lengths (total length 0.489 vs.\ 0.382), consistent with the tendancy of the Fr\'echet mean to be attracted to the origin. 
The modal tree and Fr\'echet mean are displayed in Appendix Figure~\ref{fig:experimentalSummaryTrees}. 
The posterior for $x_0$ was concentrated on two topologies: the posterior modal topology ($88.8\%$) and a topology related by a single nearest-neighbour interchange ($11.2\%$). 
This illustrates a key advantage of the Bayesian framework for estimating $x_0$: there is direct quantification of the uncertainty in the estimate, unlike the Fr\'echet mean. 
Posterior predictive sampling can be used to assess the quality of model fit, as in Appendix Figure~\ref{fig:experimentalForwardSimGDs}. 

We then estimated marginal likelihoods for three different source trees: (i) the posterior mode tree, (ii) the Fr\'echet mean and (iii) the star tree. 
When $x_0$ is the star tree, the Brownian motion transition kernel is a multiple of a Gaussian at the origin in each maximal orthant, enabling exact calculation of the marginal likelihood; $t_0$ was fixed to be the Fr\'echet variance around the star tree ($0.0325$) for this calculation. 
For the other two trees, we estimated marginal likelihoods using the Chib estimator in Section~\ref{sec:marginalLikelihood}. 
(The other estimators gave very similar values.)
In case (i), $t_0$ was fixed at its posterior mode ($0.0169$). 
In case (ii) (the Fr\'echet mean tree), we ran the inference described in Section~\ref{sec:inference} with the source tree $x_0$ fixed at the Fr\'echet mean, and used the posterior modal value for $t_0$  ($0.0167$) to estimate the marginal likelihood. 
The marginal log likelihood estimates were (i) $108.58$, (ii) $82.64$, and (iii) $-456.13$.

Clearly the star tree is a poor candidate for the source tree. 
We tested the alternative hypothesis that the posterior mode tree is the true source tree, against the null hypothesis that the Fr\'echet mean is the true source tree. 
Using the estimated marginal likelihoods~\cite{Kass1995BF}, the Bayes factor of the two hypotheses on the log scale with base $10$ is $11.3$ and we therefore concluded that there is significant evidence against the null hypothesis.
Although this represents a straightforward application for this particular data set, it illustrates the type of tests that can be performed using the marginal likelihood.

\section{Discussion}\label{sec:discussion}

The methods presented are the first that successfully fit a non-trivial parametric family of distributions to data in BHV tree space for more than a handful of taxa. 
Brownian motion kernels are analogs of Gaussian distributions in Euclidean space and hence represent a model for noise in BHV tree space. 
The ability to fit a Gaussian-type distribution to a sample of phylogenetic trees and compute the marginal likelihood opens up the possibility of a wide range of new methods in tree space. 

We have presented a basic model with data modelled as a random sample from $B(x_0,t_0)$. 
The source parameter $x_0$ has been called a \textit{diffusion mean} in other contexts  \cite{eltzner2023diffusion}, and it offers an advantage over other summary trees in BHV tree space. 
The Fr\'echet mean, for example, exhibits \textit{stickiness}, an undesirable property whereby the estimator is attracted to high-codimension strata, while the diffusion mean has been shown not to be sticky \cite{nye14diffusion}.
More importantly, unlike other approaches, our Bayesian approach to inferring $x_0$ enables direct quantification of the uncertainty in $x_0$ by inspection of the posterior sample. 
The Bayesian methods we employed brought further benefits: estimation of the marginal likelihood enables hypothesis tests for the source parameter to be performed; and posterior predictive sampling enables model checking. 
Work beyond the scope of this article has suggested that the source parameter $x_0$ is a good estimator for the species tree when the data are gene trees generated by a multispecies coalescent model \cite{degnan2009gene}. 

On the other hand, our approach has important limitations, especially in terms of scalability. 
As the number of taxa $N$ increases, computation times increase and mixing of the MCMC methods becomes poor. 
The bridge algorithm underpins all the inference procedures, and a key issue is the algorithm's performance when bridging between trees connected by geodesics which traverse high-codimension regions. 
Such data points are more likely to arise as $N$ increases and for samples of trees that have a relatively high level of dispersion. 
Simulations with $N=10$ taxa ran reliably, and simulations for $N\geq 20$ were feasible provided the dispersion $t_0$ of the simulated sample was smaller. 
The current bridge algorithm traverses regions with codimension $\geq 2$ by crudely assigning a budget of random walk steps. 
A more sophisticated proposal could account for the topological similarity between the proposed tree and the destination tree. 
It would actively direct the bridge algorithm to step round singularities towards the destination tree, thereby increasing the acceptance probability of bridge proposals. 

The methods presented here serve as a foundation for several generalizations. 
First, more complex statistical models on BHV tree space -- such as regression models or mixture models for clustering -- could be developed using Brownian kernels as analogs of Gaussian distributions in Euclidean space. 
Inference for such models would be enabled by adaptations of the bridge algorithm. 
Second, theory could be developed for a wider class of stochastic processes on BHV tree space to yield more flexible families of distributions. 
For example, a Brownian motion with non-trivial covariance structure could be considered, yielding distributions with increased dispersion in certain directions in tree space.
Third, our methods may be adapted to other stratified spaces, such as related spaces of trees. 
For example, while calculation of exact geodesics in wald space~\cite{lueg2024} is currently not possible, approximate geodesic constructions could support bridge proposals and enable inference of diffusion means. 
Finally, alternative methods for approximating intractable integrals in Bayesian inference may offer promising replacements for the bridge-based framework introduced here.

In summary, we have introduced a practical Bayesian framework for fitting Brownian motion transition kernels on BHV tree space, providing tools for uncertainty quantification, hypothesis testing, and posterior predictive checks. 
These methods are the first to scale beyond a few taxa, and form the foundation for a wider class of statistical models in the future.

\bibliographystyle{amsplain}

\providecommand{\bysame}{\leavevmode\hbox to3em{\hrulefill}\thinspace}
\providecommand{\MR}{\relax\ifhmode\unskip\space\fi MR }
\providecommand{\MRhref}[2]{%
  \href{http://www.ams.org/mathscinet-getitem?mr=#1}{#2}
}
\providecommand{\href}[2]{#2}


\newpage

\appendix
\section*{Appendices}

\renewcommand{\thesubsection}{\Alph{subsection}}
\setcounter{equation}{0}
\renewcommand{\theequation}{\thesubsection.\arabic{equation}}
\setcounter{result}{0}
\renewcommand{\theresult}{\thesubsection.\arabic{result}}
\setcounter{figure}{0}
\renewcommand{\thefigure}{\thesubsection.\arabic{figure}}
\setcounter{table}{0}
\renewcommand{\thetable}{\thesubsection.\arabic{table}}
\numberwithin{result}{subsection}
\renewcommand{\thealgorithm}{\thesubsection.\arabic{result}}

\subsection{Definition of GGF for unresolved location parameter}{\label{app:GGFunresolved}}

Suppose we have $x_0\in \BHV{N}$ which is unresolved. 
To define a GGF distribution at $x_0$, we sample a fully resolved orthant $\mathcal{O}$ which contains $x_0$ in its boundary uniformly at random, and fire a geodesic from $x_0$ in some direction within that orthant. 
The technical details are specified in the following algorithm. 

\begin{algorithm}
Suppose there are $\beta$ vertices in $x_0$ with degree greater than $3$. 
Denote these by $v_1,...,v_{\beta}$ and let $\alpha_i=\deg{v_i}-3$. 
Then there are $\Delta=\prod_{i=1}^{\beta}((2\alpha_i+1)!!)$ maximal orthants containing $x_0$ in their boundary.
A maximal orthant $\mathcal{O}$ is sampled uniformly at random from these in the following way. 
Let $\hat{x}_0=x_0$ initially. 
\begin{algorithmic}
\FORALL{$i=1,\ldots,\beta$} 
    \STATE{Set $v=v_i$ and let $W_v$ be the vertices connected to $v$.}
    \STATE{Choose a three element subset $W=\{w_1,w_2,w_3\}\subset W_v$ by sampling uniformly at random without replacement from $W_v$.}
    \STATE{Add a new vertex $u$ to the graph and an edge from $w_i$ to $u$ for $i=1,\dots,3$, denoted $e_i$.} 
    \STATE{Remove the edge from $w$ to $v$ in $\hat{x}_0$ for each vertex $w \in W_v$ and remove $v$ from $\hat{x}_0$.
    Remove the elements of $W$ from $W_v$ and set $E_v=\{e_1,e_2,e_3\}$. } 
    \WHILE{$W_v \text{ is not empty}$}
    \STATE{Choose $w$ uniformly at random without from $W_v$ and $e$ uniformly at random from $E_v$.}
    \STATE{Add a vertex $u'$ on $e$ and connect it to $w$. This creates three new edges $e'$, $e''$ and $e'''$.}
    \STATE{Remove $w$ from $W_v$ and $e$ from $E_v$.}
    \STATE{Add $e'$, $e''$ and $e''$ to $E_v$.}
    \ENDWHILE
    \ENDFOR
\end{algorithmic}
Let $\mathcal{O}$ be the topology of $\hat{x}_0$. 
(A proof by induction shows that $\mathcal{O}$ is selected uniformly at random by this procedure.)

A direction vector $\vec{u}$ in the ambient space $\R^M$ containing $\BHV{N}$ is defined by
\begin{equation*}
u_j=\begin{cases}
X_j\ \text{where}\ X_j\sim N(0,t_0)&\text{when}\ e_j\in\sigma(x_0)\\
|X_j|\ \text{where}\ X_j\sim N(0,t_0)&\text{when}\ e_j\in\sigma(\hat{x}_0)\setminus\sigma(x_0)\\
0& \text{if}\ e_j\notin\sigma(\hat{x}_0). 
\end{cases}
\end{equation*}
for $j=1,\ldots,M$ where $e_1,\ldots,e_M$ is the ordered set of splits defined in Section~\ref{sec:BHV}.
A geodesic is then extended from $x_0$ in direction $\vec{u}$, as in the case for GGF from a resolved tree, for a distance $\|\vec{u}\|$, to arrive at a random tree $y$. 
The probability density function for the corresponding distribution $\GGF(x | x_0, t_0)$ is then
\begin{equation*}
f(x| x_0, t_0) =\begin{cases} \left(\frac{1}{2}\right)^{\nu'(x,x_0)}K(x_0)\frac{1}{(2\pi)^{\N} t_0^{\N/2}}\exp-\frac{1}{2t_0}{\BHVmetric{x}{x_0}}^2 & \text{if $\Gamma'_{x,x_0}$ is simple}\\
0 & \text{otherwise,}
\end{cases}
\end{equation*}
where $\N=N-3$, $\Gamma'_{x,x_0}$ is the set $\Gamma_{x,x_0}$ with the point $x_0$ removed, $\nu'(x,x_0)$ is the number of codimension-$1$ points in $\Gamma'_{x,x_0}$, and a set is called simple if it does not contain any points on a codimension-$2$ boundary.
The factor $K(x_0)$ is
\begin{equation*}
K(x_0) = \frac{2^{\sum \alpha_i}}{\Delta}.
\end{equation*}
It accounts for (i) sampling $\mathcal{O}$ uniformly at random and (ii) use of the half-normal distribution for splits in $\sigma(\hat{x}_0)\setminus\sigma(x_0)$. 
\end{algorithm}

\subsection{Details of the MCMC bridge sampler}\label{sec:bridgeMCMCdetails}

Here we give details of Algorithm~\ref{alg:bridgesampler} which samples random walk paths $\vec{Y}_{[0,m]}$ from the conditional distribution given $Y_m=x_\star\in\BHV{N}$ and $Y_0=x_0$. 

\begin{algorithm*}[Detailed version of Algorithm 4.3]
The algorithm is initialized by running the independence proposal (Algorithm~\ref{alg:singularitybridge}) until a valid path $\vec{y}^{(0)}$ between $x_0$ and $x_\star$ is obtained. 
Then, for $j=1,2,\ldots,J$ the following steps are performed.
\begin{enumerate}
\item Generate values $a$ and $l$ by first sampling $l$ from a truncated geometric distribution on $1,\ldots,m-1$ with parameter $\alpha_b$, and then sampling $a$ uniformly on $0,\ldots,m-l-1$.
\item Sample a new bridge $\vec{y}^*$ using the partial bridge proposal with parameters $(a,l)$ conditional on $\vec{Y}=\vec{y}^{(j-1)}$, $x_0$, $t_0$. 
Calculate the proposal ratio $Q_\text{part}^{(a,l)}$ in Equation~$\eqref{equ:Qpart}$.
\item Calculate the target density ratio $P_\text{part}^{(a,l)}$ in Equation~$\eqref{equ:Ppart}$.
\item Calculate the acceptance probability $A_\text{part}^{(a,l)}$ in Equation~$\eqref{equ:Apart}$.
\item With probability $A_{\text{part}}^{(a,l)}(\vec{y}^*,\vec{y};\; x_\star, x_0, t_0)$ set $\vec{y}^{(j)}=\vec{y}^*$; otherwise set $\vec{y}^{(j)}=\vec{y}$. 
\end{enumerate}
Output the sample of bridges $\vec{y}^{(1)},\ldots,\vec{y}^{(J)}$. 
\end{algorithm*}

We need to determine the proposal density ratios and target density ratios required by the algorithm. 
The probability density function for the independence proposal (Algorithm~\ref{alg:singularitybridge}) is
\begin{equation}\label{equ:qind}
q_{\text{ind}}(\vec{y}\;|\;y_0=x_0,y_m=x_\star,t_0) = \prod_{j=1}^{m-1}q_{\text{ind}}(y_j\;|\;y_{j-1},x_\star,t_0)
\end{equation} 
where
\begin{equation*}
q_{\text{ind}}(y_j\;|\;y_{j-1},x_\star,t_0) = w(\mu_j)\GGFdens\left(y_j\;|\;\mu_j,\varjm{j}{m}\right)+[1-w(\mu_j)]\GGFdens\left(y_j\;|\;y_{j-1},\frac{t_{0}}{m}\right),
\end{equation*}
$w(x)$ is defined in Equation~$\eqref{equ:weightfunction}$ and  $\mu_j$, $\varjm{j}{m}$ are defined in Algorithm~\ref{alg:singularitybridge}. 

Now suppose $\vec{Y}^*_{[0,m]}=\vec{y}^*$ is proposed from $\vec{Y}_{[0,m]}=\vec{y}$ via the partial bridge proposal parameters $(a,l)$, so that $y^*_i=y_i$ for $i\leq a$ and for $i\geq a+l+1$. 
The probability density function for the proposal is
\begin{equation}\label{equ:partialbridgepropdens}
q_{\text{part}}^{(a,l)}(\vec{y}^*\;|\;\vec{y},t_0) = \prod_{j=a+1}^{a+l}q_{\text{part}}(y^*_j\;|\;y^*_{j-1},y_{a+l+1},t_0)
\end{equation} 
where
\begin{multline}\label{equ:qpart}
q_{\text{part}}^{(a,l)}(y^*_j\;|\;y^*_{j-1},y_{a+l+1},t_0) = w\left(\mu_j^*\right)\GGFdens\left(y^*_j\;\middle|\;\mu^*_j,\varstar\right)\\ +\left[1-w\left(\mu_j^*\right)\right]\GGFdens\left(y^*_j\;\middle|\;y^*_{j-1},{t_{0}}/{m}\right).
\end{multline}
Terms in these equations are defined as follows. 
First $\gamma$ is the geodesic segment $\Gamma_{y^*_{j-1},y_{a+l+1}}[0,1/(a+l+2-j-p^*_j)]$ where
\begin{equation*}
p_j^* = f_p\left( \Gamma_{y^*_{j-1},y_{a+l+1}} \right).
\end{equation*}
As for Algorithm~\ref{alg:singularitybridge}, if there is no boundary with codimension greater than 1 in $\gamma$ then $\mu_j^*$ is set to be $\Gamma_{y^*_{j-1},y_{a+l+1}}(1/(a+l+2-j-p^*_j))$; otherwise $\mu_j^*$ is the point on $\gamma$ with codimension $\geq 2$ closest to $y^*_{j-1}$.
Finally, $\varstar$ is
\begin{equation*}
\varstar = \frac{l+1-(j-a)}{l+2-(j-a)}\frac{t_0}{m}.
\end{equation*}

The proposal ratio is therefore
\begin{equation}\label{equ:Qpart}
Q_{\text{part}}^{(a,l)}(\vec{y}^*,\vec{y};\; t_0)=\prod_{j=a+1}^{a+l} \frac{q_{\text{part}}(y_j\;|\;y_{j-1},y_{a+l+1},t_0)}{q_{\text{part}}(y^*_j\;|\;y^*_{j-1},y_{a+l+1},t_0)}.
\end{equation}
Using Equation~$\eqref{equ:bridgedensity}$, the target distribution density ratio is
\begin{align}\label{equ:Ppart}
P_{\text{part}}^{(a,l)}(\vec{y}^*,\vec{y};\; x_\star, x_0, t_0) &= \frac{f_{\{\vec{Y}_{[1,m-1]}|Y_m\}}\left((y^*_1,\ldots,y^*_{m-1}) \;|\; y^*_m=x_\star,\, x_0,\, t_0\right)}{f_{\{\vec{Y}_{[1,m-1]}|Y_m\}}\left((y_1,\ldots,y_{m-1}) \;|\; y_m=x_\star,\, x_0,\, t_0\right)} \notag\\
&= \prod_{j=a+1}^{a+l+1} \frac{\GGFdens(y_j^*\;|\;y_{j-1}^*,t_0/m)}{ \GGFdens(y_j\;|\;y_{j-1},t_0/m)}.
\end{align}
The acceptance probability for the proposed path $\vec{y}^*$ given $\vec{y}$ is 
\begin{equation}\label{equ:Apart}
A_{\text{part}}^{(a,l)}(\vec{y}^*,\vec{y};\; x_\star, x_0, t_0) = \min\left\{ 1, P_{\text{part}}^{(a,l)}(\vec{y}^*,\vec{y};\; x_\star, x_0, t_0)Q_{\text{part}}^{(a,l)}(\vec{y}^*,\vec{y};\; t_0) \right\}.
\end{equation}

\begin{lemma}\label{lem:bridgeMCMCconv}
The Markov chain induced by Algorithm~\ref{alg:bridgesampler} almost surely converges to its stationary distribution, which is the conditional distribution with density function defined in Equation~$\eqref{equ:bridgedensity}$.
\end{lemma}

\begin{proof}
The Markov chain $(\vec{Y}^{(j)})_{j\in \mathbb{N}}$ induced by Algorithm~\ref{alg:bridgesampler} has the conditional distribution in Equation~$\eqref{equ:bridgedensity}$ as its stationary distribution by the construction of the Metropolis-Hastings steps. 

The Borel volume measure on $\BHV{N}$ was defined in \cite{willis_confidence_2016}.
Let $\nu$ be the Borel volume measure on the product
\begin{equation*}
\BHV{N}^{{(m-1)}} = \BHV{N}\times \BHV{N}\times\cdots\times\BHV{N},\quad\text{($m-1$ terms),}
\end{equation*}
where the product Borel $\sigma$-algebra is denoted $\mathcal{B}(\BHV{N}^{{(m-1)}})$. 
Let 
\begin{equation*}
S=\{\vec{y}_{[1,m-1]}: \Gamma_{y_{i-1},y_i} \text{ is simple for }i=1,...,m;\; y_0=x_0, y_m=x_\star\}.
\end{equation*}
Consider the measure defined by 
$$
\nu'(A)=\nu(A \cap S) \text{ for } A \in \mathcal{B}(\BHV{N}^{{(m-1)}}).
$$
Since there a non-zero probability of generating $l=m-1$ from the truncated geometric distribution in Algorithm~\ref{alg:bridgesampler} and we have positive independence proposal density for any valid bridge $\mathbf{Y}^*$, then there is a positive probability of moving into $A$ in one step for any $A$ with $\nu'(A)>0$ for any valid starting bridge $\mathbf{Y}$. Therefore the Markov chain $(\vec{Y}^{(j)})_{j\in \mathbb{N}}$ is $\varphi$-irreducible with respect to $\nu'$. It therefore converges $\nu'$-almost surely to its stationary distribution by for example Proposition 1 in~\cite{rosenthal2001review}.
\end{proof}

\subsection{Details of the acceptance probability for $x_0$ and $t_0$ proposals}\label{sec:mainMCMCdetails}

Here we give details of the acceptance probabilities for the proposals used in the MCMC scheme described in Section~\ref{sec:inference} which samples from the posterior for $x_0,t_0$. 
First consider the proposal for $x_0$ given the value $l$ of the number of bridge steps to resample. 
Suppose bridges $\vec{Y}_i^*=\vec{y}_i^*$ are proposed from $\vec{Y}_i=\vec{y}_i$, $i=1,\ldots,n$, so that $y^*_{i,j}=y_{i,j}$ for $j>l$ and for all $i$. 
We fix the convention $y_{i,0}^*=x_0^*$ and $y_{i,0}=x_0$ for all $i$.  
The probability density function for the proposal is
\begin{equation*}
\begin{split}
q^{(l)}_{\text{source}}(x_0^*,\vec{y}_1^*,\ldots,\vec{y}_n^*\;|\;x_0,\vec{y}_1,\ldots,\vec{y}_n,t_0)=
\GGFdens(x_0^*|x_0,\lambda_0^2)\\
\times
\prod_{i=1}^n\prod_{j=1}^l q^{(0,l)}_{\text{part}}\left( y_{i,j}^*\;|\;y_{i,j-1}^*,y_{i,l+1},t_0 \right).
\end{split}
\end{equation*}
Here $q^{(0,l)}_{\text{part}}$ is defined by Equation~$\eqref{equ:qpart}$. 
If $l=0$ then the empty product is taken to be $1$. 
The proposal ratio is
\begin{equation}\label{equ:Qx0}
Q^{(l)}_{\text{source}}(x_0^*,\vec{y}_1^*,\ldots,\vec{y}_n^*,x_0,\vec{y}_1,\ldots,\vec{y}_n;\;t_0) = 
\prod_{i=1}^n\prod_{j=1}^l\frac{ q^{(0,l)}_{\text{part}}\left( y_{i,j}\;|\;y_{i,j-1},y_{i,l+1},t_0 \right) }{ q^{(0,l)}_{\text{part}}\left( y_{i,j}^*\;|\;y_{i,j-1}^*,y_{i,l+1},t_0 \right) }.
\end{equation}
The term $\GGFdens(x_0^*|x_0,\lambda_0^2)$ is unchanged if $x_0,x_0^*$ are swapped, and so those terms cancel in the proposal ratio. 
Using Equation~$\eqref{equ:maintarget}$, the target distribution density ratio is
\begin{equation*}
P^{(l)}_{\text{source}}(x_0^*,\vec{y}_1^*,\ldots,\vec{y}_n^*,x_0,\vec{y}_1,\ldots,\vec{y}_n;\;t_0) = 
\frac{\pi(x_0^*,t_0)}{\pi(x_0,t_0)}\prod_{i=1}^n \prod_{j=1}^{l+1}\frac{ \GGFdens(y_{i,j}^*\;|\;y_{i,j-1}^*,t_0/m) }{ \GGFdens(y_{i,j}\;|\;y_{i,j-1},t_0/m) }.
\end{equation*}
The acceptance probability for the proposal is $\min\{1,A^{(l)}_{\text{source}}\}$ where
\begin{multline*}
A^{(l)}_{\text{source}} = \frac{\pi(x_0^*,t_0)}{\pi(x_0,t_0)}\prod_{i=1}^n \frac{ \GGFdens(y_{i,l+1}\;|\;y_{i,l}^*,t_0/m) }{ \GGFdens(y_{i,l+1}\;|\;y_{i,l},t_0/m) }\\ \times\prod_{j=1}^{l}\frac{ \GGFdens(y_{i,j}^*\;|\;y_{i,j-1}^*,t_0/m) q^{(0,l)}_{\text{part}}\left( y_{i,j}\;|\;y_{i,j-1},y_{i,l+1},t_0 \right) }{ \GGFdens(y_{i,j}\;|\;y_{i,j-1},t_0/m) q^{(0,l)}_{\text{part}}\left( y_{i,j}^*\;|\;y_{i,j-1}^*,y_{i,l+1},t_0 \right)}.
\end{multline*}

The proposal ratio for the $t_0$ proposal is $t_0^*/t_0$. 
The target density ratio is
\begin{equation}\label{eq:t0AcceptanceRatio}
P_{\text{disp}}(t_0^*\;|\;t_0,\vec{y}_1,\ldots,\vec{y}_n) = \frac{\pi(x_0,t_0^*)}{\pi(x_0,t_0)}\prod_{i=1}^n \prod_{j=1}^{m}\frac{ \GGFdens(y_{i,j}\;|\; y_{i,j-1},t_0^*/m) }{\GGFdens(y_{i,j}\;|\; y_{i,j-1},t_0/m)}. 
\end{equation}
The acceptance probability for the proposal is $\min\{1,(t_0^*/t_0)P_{\text{disp}}(t_0^*\;|\;t_0,\vec{y}_1,\ldots,\vec{y}_n)\}$.

\subsection{Algorithms for the marginal likelihood}

The following algorithms require samples to be drawn for the bridges $\vec{y}_i$, $i=1,\ldots,n$, conditional on $x_0, t_0$ and the data $\{x_i\}$. 
This is achieved in the same way as the MCMC scheme in Section~\ref{sec:inference} but dropping the proposals for $x_0$ and $t_0$. 
Since $t_0$ is assumed to be fixed and known, in this section we will suppress notational dependence on $t_0$ unless absolutely necessary. 
We will simplify the notation from Equation~\eqref{equ:RW-like} by writing $f(\vec{y} \;|\; x_0)$ for $f_{\vec{Y}_{[1,m]}}(\vec{y}_{[1,m]} \;|\; x_0, t_0)$. 
We will additionally use the notation $\oldvec{\vec{y}}=(\vec{y}_1,\ldots,\vec{y}_n)$ for a set of bridges between $x_0$ and the data points $(x_1,\ldots,x_n)$. 

\begin{algorithm}[\textsc{Chib estimate}]\label{alg:ML:Chib}
\begin{algorithmic}
    Fix a value $M_1 \in \mathbb{N}$ for the number of samples from the conditional posterior distribution, a value $M_2 \in \mathbb{N}$ for the number of samples from the independence proposal and a value $h \in \mathbb{N}$ for the number of points in the conditional posterior sample at which to estimate the conditional posterior density. Set $H=\left \lfloor{\frac{M_1}{h}}\right \rfloor$. Sample $\multfixedbridges{y}{1},\ldots,\multfixedbridges{y}{M_1}$ from the conditional posterior distribution.
\FOR{$i=1,\ldots,n$}
\FOR{$j=1,\ldots M_2$}
\STATE{Simulate a bridge $\vec{w}_{i}^{(j)}$ from the independence proposal using Algorithm~\ref{alg:singularitybridge}}.
\ENDFOR
\FOR{$k=1,\ldots,h$}
\STATE{Set $\vec{y}_i^*=\vec{y}_i^{(Hk)}$}
\STATE{Calculate an estimate $\hat{f}_{C}^{(k)}(\vec{y}_i^*\;|\;x_0,x_i)$ of the conditional density in Equation~\eqref{equ:bridgedensity} by
        \begin{align*}
    \hat{f}_{C}^{(k)}(\vec{y}_i^*\;|\;x_0,x_i)
=\frac{M_1^{-1}\sum_{j=1}^{M_1}A_{\text{ind}}(\vec{y}_i^*,\vec{y}_i^{(j)};x_0,x_i)}{M_2^{-1}\sum_{j=1}^{M_2} A_{\text{ind}}(\vec{w}_i^{(j)},\vec{y}_i^*;x_0,x_i)}.
\end{align*}
and then an estimate of the log marginal likelihood by
\begin{equation*}\label{eq:ML:ChibMLEst}
    \logmle_{C}^{(k)}(x_i\;|\;x_0)=\log f(\vec{y}_i^*\;|\;x_0)-\log\hat{f}_{C}^{(k)}(\vec{y}_i^*\;|\;x_0,x_i)-q_{\text{ind}}(\vec{y}_i^*\;|\;x_0,x_i).
\end{equation*}}
\STATE{Here, $A_\text{ind}$ is defined by
$$A_{\text{ind}}(\vec{y}_i^*,\vec{y}_i;\; x_i, x_0) = \min\left\{ 1, P_{\text{ind}}(\vec{y}_i^*,\vec{y}_i;\; x_i, x_0)Q_{\text{ind}}(\vec{y}_i^*,\vec{y}_i) \right\}.
$$
where
$$
Q_{\text{ind}}(\vec{y}_i^*,\vec{y}_i) =\frac{q_{\text{ind}}(\vec{y}_i\;|\;x_0,x_i)}{q_{\text{ind}}(\vec{y}_i^*\;|\;x_0,x_i)}
$$
and 
$$
P_{\text{ind}}(\vec{y}_i^*,\vec{y}_i;\; x_i, x_0) = \begin{cases}
    \frac{f(\vec{y}_i^*\;|\;x_0)}{f(\vec{y}_i\;|\;x_0)} & \text{if }f(\vec{y}_i\;|\;x_0)>0, \\
    0 & \text{otherwise.}
\end{cases}$$}
\ENDFOR
\STATE{Set $\logmle_{C}(x_i\;|\;x_0)=\log\left(\frac{1}{h}\sum_{k=1}^h\exp\logmle_{C}^{(k)}(x_i\;|\;x_0)\right)$.}
\ENDFOR
\end{algorithmic}
Output the estimated log marginal likelihood $\logmle_{C}(\vec{x}\;|\;x_0)$ given by
$$
\logmle_{C}(\vec{x}\;|\;x_0)=\sum_{i=1}^n \logmle_{C}(x_i\;|\;x_0).
$$
\end{algorithm}



\textbf{Tunnel method.} 
The estimated marginal likelihood is calculated as the ratio of the normalising constants of two probability density functions: (i) the density function of the conditional distribution of the bridge $\vec{y}_{[1,m-1]}$ between $x_0$ and $x_\star$, and (Equation~$\eqref{equ:bridgedensity}$) (ii) the density function of some normalised reference distribution. 
In our case, the obvious candidate for the reference distribution is the independence proposal distribution for the bridge. 
As for the Chib method, the estimate is obtained using samples from both the conditional distribution and the independence proposal distribution. 
We note that both the tunnel estimator and Chib estimator can be computed from the same sets of sampled bridges. 
To improve numerical stability of the estimator the method of~\cite{GronauQuentinF.2017Atob} is adopted.

\begin{algorithm}[\textsc{Tunnel sampling}]\label{alg:ML:tunnel}
\begin{algorithmic}
    Fix a value $M_1 \in \mathbb{N}$ for the number of samples from the conditional posterior distribution, a value $M_2 \in \mathbb{N}$ for the number of samples from the independence proposal and a value $K \in \mathbb{N}$ for the number of iterations when calculating the estimate. 
Simulate a sample $\multfixedbridges{y}{1},\ldots,\multfixedbridges{y}{M_1}$ from the conditional posterior distribution.
\STATE{Set $c_1=\frac{M_1}{M_1+M_2}$ and $c_2=\frac{M_2}{M_1+M_2}$.}
\FOR{$i=1,\ldots,n$}
\FOR{$j=1,\ldots M_1$}
\STATE{Calculate $l^{(j)}_{f,i}$ by
    \begin{equation*}
l^{(j)}_{f,i}=\log f(\vec{y}_i^{(j)}\;|\;x_0)-\log q_{\text{ind}}(\vec{y}_i^{(j)}\;|\;x_0,x_i).
\end{equation*}}
\ENDFOR
\STATE{Set $l_i$ to be the median of the set $ \{l^{(j)}_{f,i}:j=1,\ldots,M_1\}$}.
\FOR{$j=1,\ldots M_2$}
\STATE{Simulate a bridge $\vec{w}_{i}^{(j)}$ from the independence proposal using Algorithm~\ref{alg:singularitybridge}}.
\STATE{Calculate $l^{(j)}_{q,i}$ by
    \begin{equation*}
l^{(j)}_{q,i}=\log f(\vec{w}_i^{(j)}\;|\;x_0)-\log q_{\text{ind}}(\vec{w}_i^{(j)}\;|\;x_0,x_i).
\end{equation*}
}
\ENDFOR
\STATE{Set $\hat{f}_{TS}(x_i|x_0)^{\left[0\right]}=0.1$.}
\FOR{$k=1,\ldots,K$}
\STATE{Calculate $\hat{f}_{T}(x_i|x_0)^{\left[k\right]}$ by
        \begin{align*}
   \hat{f}_{T}(x_i\;|\;x_0)^{\left[k\right]}=&\frac{1}{M_2}\sum_{j=1}^{M_2}\frac{\exp( l_{q,i}^{(j)}-l_i)}{c_2\exp( l_{q,i}^{(j)}-l_i)+c_1\hat{f}_{T}(x_i|x_0)^{\left[k-1\right]}}\\ \times&\left(\frac{1}{M_1}\sum_{j=1}^{M_1}\frac{1}{c_2 \exp( l_{f,i}^{(j)}-l_i)+c_1\hat{f}_{T}(x_i|x_0)^{\left[k-1\right]}}\right)^{-1}.   
\end{align*}}
\ENDFOR
\STATE{Set $\logmle_{T}(x_i\;|\;x_0)=\log(\hat{f}_{T}(x_i\;|\;x_0)^{\left[K\right]}+l_i).$}
\ENDFOR
\end{algorithmic}
Output the estimated log marginal likelihood $\logmle_{T}(\vec{x}\;|\;x_0)$ given by
$$
\logmle_{T}(\vec{x}\;|\;x_0)=\sum_{i=1}^n \logmle_{T}(x_i\;|\;x_0).
$$
\end{algorithm}

\textbf{Stepping stone method.} 
This is the method of (generalised) stepping stone sampling from~\cite{FanYu2011Capm} which is a generalisation of the method in~\cite{XieWangang2011IMLE}. 
In a similar manner to the tunnel method, an estimate is calculated of the ratio of the normalising constants of the conditional bridge distribution and the independence proposal distribution. The stepping stone estimator requires samples from a number of distributions that are on a path between the conditional bridge distribution and the independence proposal distribution. These distributions have the following unnormalised density function for different values of $\beta\in[0,1]$:
\begin{equation}\label{equ:stepStoneDists}
     \lambda_\beta(\vec{y};x_0,x_*)=f_{\{\vec{Y}_{[1,m-1]}|Y_m\}}\left(\vec{y}_{[1,m-1]} \;|\;  x_0,\, y_m=x_\star\right) ^\beta q_{\text{ind}}(\vec{y}_{[1,m-1]}\;|\;x_0,\,y_m=x_\star)^{1-\beta}
\end{equation}
where $q_{\text{ind}}$ is the probability density function of the independence proposal defined in Equation~$\eqref{equ:qind}$. 
We denote by $F_\beta$ the distribution defined by the unnormalised density in Equation~\eqref{equ:stepStoneDists}. We choose some value $K \in \mathbb{N}$ and values $0=\beta_0<\beta_1\ldots <\beta_{K}=1$ and generate samples from $F_{\beta_k}$ for $k=0,\ldots,K-1$. $F_{\beta_0}$ is the independence proposal distribution and $F_{\beta_K}$ is the conditional bridge distribution. Simulation from $F_{\beta_k}$ is achieved by using the approach specified in Algorithm~\ref{alg:bridgesampler} with a modification to the target distribution density ratio.

Samples are obtained by what~\cite{LartillotNicolas2006CBFU} called the quasistatic method, which means that the last bridge sampled from $F_{\beta_{k-1}}$ is passed in as the starting point of the MCMC chain when sampling from $F_{\beta_{k}}$, instead of having a burn-in period for $K$ Markov chains. It is much more efficient to simulate from the independence distribution directly, rather than using MCMC. A burn-in period is therefore required for the distribution $F_{\beta_{1}}$.
\begin{algorithm}[\textsc{Stepping stone estimate}]\label{alg:ML:StepStone}
\begin{algorithmic}
\STATE{Fix a number $K\in \mathbb{N}$ and values $0 =\beta_0 < \beta_1 <\ldots <\beta_K=1 $. Fix a number $M_0\in \mathbb{N}$ for the number of bridges to simulate under the independence proposal and a number $M \in \mathbb{N}$ for the number of samples to be simulated by the MCMC for each $\beta_k$. Fix a number $b \in \mathbb{N}$ for number of burn-in iterations and $c \in \mathbb{N}$ for the number of thin iterations to be used in the MCMC.}
\FOR{$k=1\ldots,K$}
\IF{$k=1$}
\FOR{$i=1,\ldots,n$}
\FOR{$j=1,\ldots M_0$}
\STATE{Simulate a bridge $\vec{y}_{i}^{(j,1)}$ from the independence proposal using Algorithm~\ref{alg:singularitybridge}}.
\ENDFOR
\STATE{Repeatedly run the independence proposal until a valid bridge $\vec{y}_i^{\text{start}}$ between $x_0$ and the data point $x_i$ is produced.} 
\ENDFOR
\STATE{Set $\multfixedbridges{y}{\text{start}}=(\vec{y}_1^{\text{start}},\dots,\vec{y}_n^{\text{start}})$}
\ELSE
\STATE{Sample $\multfixedbridges{y}{1,k},\ldots,\multfixedbridges{y}{M,k}$ using Algorithm~\ref{alg:ML:StepStoneSim} with $\beta=\beta_{k-1}$, $c=c$, $M=M$ and $\multfixedbridges{y}{0}=\multfixedbridges{y}{\text{start}}$. If $k=2$ use $b=b$ and otherwise use $b=0$.}
\STATE{Set $\multfixedbridges{y}{\text{start}}=\multfixedbridges{y}{M,k}$.}
\ENDIF
\FOR{$i=1,\dots,n$}
\STATE{Calculate $\eta_i^{(k)}$ as
            \begin{equation*}
   \eta_i^{(k)}=\max_{j=1,\ldots,M} \left[\frac{f(\vec{y}_i^{(
   j,k)}\;|\;x_0)}{q_{\text{ind}}(\vec{y}_i^{(j,k)}\;|\;x_0,x_i)}\right] .
\end{equation*}}
\STATE{Calculate $\logmle_S(x_i|x_0)^{(k)}$ as
            \begin{equation*}
\logmle_S(x_i|x_0)^{(k)}=(\beta_k-\beta_{k-1})\log\eta_k + \log\left(\sum_{j=1}^M\left[\frac{f(\vec{y}_i^{(j,k)}\;|\;x_0)}{\eta_i^{(k)}q_{\text{ind}}(\vec{y}_i^{(j,k)}\;|\;x_0,x_i)}\right]^{(\beta_k-\beta_{k-1})}\right).
\end{equation*}}
\ENDFOR
\STATE{Calculate the estimate $\logmle_{S}(x_i\;|\;x_0)$ for $\log\walkdens(x_i\;|\;x_0)$ by
        \begin{equation*}
    \logmle_{S}(x_i\;|\;x_0)=\sum_{k=1}^K\logmle_S(x_i\;|\;x_0)^{(k)}.
\end{equation*}}
\ENDFOR
 \STATE{Output the estimated log marginal likelihood $\logmle_{S}(\vec{x}\;|\;x_0)$ given by
 $$
 \logmle_{S}(\vec{x}\;|\;x_0) = \sum_{i=1}^n \logmle_{S}(x_i\;|\;x_0).
 $$
 }
\end{algorithmic}
\end{algorithm}

In practice, we will use a value of $K=100$ and equally spaced points $\beta_k=\frac{k}{K}$, which is the approach adopted in \cite{FanYu2011Capm}. \cite{XieWangang2011IMLE} suggest a different spacing of the points $\beta_k$ that places more points near to $\beta_0=0$, when the reference distribution is the prior. In our case, as the reference distribution is carefully constructed to be similar to the posterior, equally spaced points should suffice.

Consider the following unnormalised density function
$$
\lambda_\beta(\hat{\vec{y}};x_0,\vec{x})=\prod_{i=1}^{n}\lambda_\beta(\vec{y}_i;x_0,x_i),
$$ 
where $\lambda_\beta(\vec{y}_i;x_0,x_i)$ is given by Equation~\ref{equ:stepStoneDists}. The following Algorithm specifies our approach to simulating from a distribution with such an unnormalised density and is used to simulate the samples required by Algorithm~\ref{alg:ML:StepStone}.
\begin{algorithm}[\textsc{Stepping stone sampling}]\label{alg:ML:StepStoneSim}
\begin{algorithmic}
\STATE{Input $\beta\in[0,1]$, the number of burn-in iterations $b>0$, the number of thin iterations $c>0$, the number of samples to be outputted $M>0$, and a set $\vec{\hat{y}}^{(0)}$ of bridges.}
\FOR{$j=1,\ldots,Mc+b$}
\FOR{$i=1,\ldots,n$}
\STATE{Sample a new bridge $\vec{y}_i^*$ using the partial bridge proposal conditional on $\vec{Y}_i=\vec{y}_i^{(j-1)}$ and $x_0$. Calculate the proposal ratio $Q_\text{part}$ in Equation~$\eqref{equ:Qpart}$.}
\STATE{Calculate the target density ratio $P_\text{step}$ given by 
\begin{equation*}
P_{\text{step}}(\vec{y}_i^*,\vec{y}_i;\; x_i, x_0,\beta) = \left[\prod_{k=a+1}^{a+l+1} \frac{\GGFdens(y_{i,k}^*\;|\;y_{i,k-1}^*,t_0/m)}{ \GGFdens(y_{i,k}\;|\;y_{i,k-1},t_0/m)}\right]^\beta \left[\frac{q_{\text{ind}}(\vec{y}_i^*|x_0,x_i)}{q_{\text{ind}}(\vec{y}_i|x_0,x_i)}\right]^{1-\beta}.
\end{equation*}}
\STATE{Calculate the acceptance probability $A_\text{step}$ by
$$
A_{\text{step}}(\vec{y}_i^*,\vec{y}_i;\; x_i, x_0, \beta) = \min\left\{ 1, P_{\text{step}}(\vec{y}_i^*,\vec{y}_i;\; x_i, x_0, \beta)Q_{\text{part}}(\vec{y}_i^*,\vec{y}_i) \right\}
$$}
\STATE{With probability $A_{\text{part}}(\vec{y}_i^*,\vec{y}_i;\; x_i, x_0, \beta)$ set $\vec{y}_i^{(j)}=\vec{y}_i^*$; otherwise set $\vec{y}^{(j)}_i=\vec{y}_i$. 
}
\ENDFOR
\STATE{Set $\multfixedbridges{y}{j}=(\vec{y}_1^{(j)},\ldots \vec{y}_n^{(j)})$}
\ENDFOR
\STATE{Refine the sample of sets of bridges $\multfixedbridges{y}{1},\ldots,\multfixedbridges{y}{
Mc+b}$ to include only indices $jc+b$ and reindex by the map $(jc+b) \mapsto j$.}
\STATE{Output the sample of sets of bridges $\multfixedbridges{y}{1},\ldots,\multfixedbridges{y}{M}$.}
\end{algorithmic}
\end{algorithm}

\subsection{Proof of Bayesian consistency (Theorem~\ref{thm:consistency})}\label{sec:consistencyproof}

Recall that $\mathcal{A}$ is the Borel $\sigma$-algebra on $\BHV{N}^{(0)}$, which is the subset of fully resolved trees.
For simplicity, we assume that with prior probability 1, $X_0\in\BHV{N}^{(0)}$, although the proof can be made more general to accommodate other priors. 
We aim show the following conditions hold. 
\begin{description}
     \item[Condition I] The function $x_0 \mapsto B(x_0,t_0)(A)$ is measurable for all $A \in \mathcal{A}$.
    \item[Condition II] $x_0 \neq x_0' \implies B(x_0,t_0)\neq B(x_0',t_0)$.
\end{description}
Then Theorem~\ref{thm:consistency} follows as a consequence of Theorem 2.4 in~\cite{MillerJeffreyW2018Adto}. 

\textbf{Proof of Condition I.\ } If the function
\begin{equation}\label{equ:measureablewalk}
x_0 \mapsto \Wm{x_0}{t_0}(A)
\end{equation}
is measurable for all $A \in \mathcal{A}$ we have measurability of $x_0 \mapsto \BB{x_0}{t_0}(A)$ for all $A \in \mathcal{A}$, since the pointwise limit of a sequence of measurable functions is measurable. 
The proof in \cite{nye2020random} established weak convergence of the random walk distributions $\Wm{x_0}{t_0}$ to $\BB{x_0}{t_0}$. 
This guarantees the convergence of $\Wm{x_0}{t_0}(A) \to \BB{x_0}{t_0}(A)$ as $m \to \infty$ for all sets $A \in \mathcal{A}$ with $\BB{x_0}{t_0}(\delta A)=0$, where $\delta A$ is the boundary of $A$. This only presents a problem for us if $\BB{x_0}{t_0}(A)>0$ for some lower dimensional subset of $\BHV{N}$, which is not the case as $\BB{x_0}{t_0}$ is absolutely continuous with respect to the Borel volume measure on $\BHV{N}^{(0)}$ (see Definition 8 of \cite{nye2020random}).

We use an induction argument to show that the function in Equation~$\eqref{equ:measureablewalk}$ is measurable.
We show in Lemma \ref{lem:cons:GGFMeas} below that $(x_0,x) \mapsto \GGFdens(x \;|\; x_0,t_0)$ is measurable with respect to the product $\sigma$-algebra $\mathcal{A}\bigotimes \mathcal{A}$ on $\BHV{N}^{(0)} \times \BHV{N}^{(0)}$. 
Then by Fubini's theorem, $x_0 \mapsto W_{\GGF}(x_0,t_0;1)(A)$ is a measurable function for all $A\in \mathcal{A}$.

Now assume, for some $m>1$, that $f_{W}(x \;|\; x_0,t;m-1)$ is measurable as a function of $(x_0,x)$ for any $t>0$. 
We write 
\begin{equation}\label{eq:cons:RWMeasurability}
   f_W(x \;|\; x_0,t;m) = \int_{\BHV{N}} f_{W}\left(y\;|\;x_0,t^{(m)};m-1\right)\,\GGFdens\left(x\;|\;y,t/m\right)\,dy
\end{equation}
where $t^{(m)}=t(m-1)/m$. 
By assumption $f_{W}(y \;|\; x_0,t^{(m)};m-1)$ is measurable as a function of $(x_0,y)$ and by Lemma~\ref{lem:cons:GGFMeas}, $\GGFdens(x\;|\;y,t/m)$ is measurable as a function of $(y,x)$. 
It follows that 
$$
(x_0,y,x) \mapsto f_{W}\left(y\;|\;x_0,t^{(m)};m-1\right)\,\GGFdens\left(x\;|\;y,t/m\right)
$$
is a measurable function with respect to the $\sigma$-algebra $\mathcal{A}\bigotimes\mathcal{A}\bigotimes\mathcal{A}$.
Using Fubini's theorem and Equation~\eqref{eq:cons:RWMeasurability} we see that 
$$(x_0,x) \mapsto f_W(x \;|\; x_0,t;m)$$
is a measurable function for all $t>0$, and by induction this holds for all $m$.
Finally, the function
$$
x_0 \mapsto\Wm{x_0}{t_0}(A) = \int_A f_W(x \;|\; x_0,t;m)\; dx
$$
is measurable by Fubini's theorem.

\textbf{Proof of Condition II.\ }
We aim to show $x_0 \neq x_0' \implies B(x_0,t_0)\neq B(x_0',t_0)$, and we do this for two specific cases: (i) when $x_0$ and $x_0'$ are at different distances from the origin and (ii) when $x_0$ and $x_0'$ lie in different maximal orthants but are the same distance from the origin. 
For brevity, we omit the proof for the remaining case, when $x_0$ and $x_0'$ lie in the same maximal orthant and at the same distance from the origin, which is similar to case (ii).

In both cases (i) and (ii), the proof relies on a construction from \cite{nye2020random}.
There, the probability measure $\BB{x_0}{t_0}$ on $\BHV{N}$ was defined via a projection map $\mathcal{P}$ that takes paths on $\BHV{N}$ starting at $x_0$ which avoid codimension-2 to paths on $\mathbb{R}^{N'}_{\geq 0}$, where $N'=N-3$. 
The projection map is used to establish both conditions above. 
It operates via a series of reflections as follows. 
Suppose $x_0$ lies in the interior of a maximal orthant and let $\eta:[0,t_0]\rightarrow\BHV{N}$ denote a Brownian sample path starting from $x_0$. 
It was shown in \cite{nye2020random} that $\eta$ almost surely traverses a finite sequence of distinct maximal orthants. 
Since $\eta$ avoids codimension-2 almost surely, at most 1 split is replaced in $\eta(t)$ each time it hits a codimension-1 boundary in $\BHV{N}$. 
This sets up a sequence of isometries between the closure of each maximal orthant traversed by $\eta$ and $\R^{N'}_{\geq 0}$, under which the image $\eta$ is a sample path of reflected Brownian motion on $\R^{N'}_{\geq 0}$. 
More details are given in \cite{nye2020random}.

A consequence of the projection map is that, for any $r>0$, the probability that a Brownian motion starting from $x_0 \in \BHV{N}$ lies is in $K(0,r) \subset \BHV{N}$ at time $t_0$, is the same as the probability that a Brownian motion in $\mathbb{R}^{N'}$ starting from a distance $d(x_0,0)$ from the origin is in $K(0,r) \subset \mathbb{R}^{N'}$ at time $t_0$. 
This proves the result for case (i). We will prove the result for case (ii) in Lemma~\ref{lem:cons:IdentDiffOrthants}.


\qed

The following lemmas were used in the proof of Theorem~\ref{thm:consistency}.

\begin{lemma}\label{lem:cons:GGFMeas}
The function $(x_0,x) \mapsto \GGFdens(x\;|\;x_0,t)$ is measurable with respect to the product $\sigma$-algebra $\mathcal{A}\bigotimes \mathcal{A}$ on $\BHV{N}^{(0)} \times \BHV{N}^{(0)}$ for any $t>0$.
\end{lemma}

\begin{proof}
Let $\mathcal{G} = \{(x_0,x) \in \BHV{N}^{(0)} \times \BHV{N}^{(0)} : \Gamma_{x_0,x} \text{ is simple}\}$, and let $I_{(x_0,x)\in \mathcal{G}}$ be the corresponding indicator function on $\BHV{N}^{(0)} \times \BHV{N}^{(0)}$.  
We rewrite the density function for GGF centred at $x_0$ with dispersion $t$ from Equation~\eqref{equ:GGF-dens}, as
$$
\GGFdens(x \;|\; x_0,t) = I_{(x_0,x)\in \mathcal{G}}\left(\frac{1}{2}\right)^{\nu(x_0,x)}\frac{1}{(2\pi)^\N t^{N'/2}}\exp-\frac{1}{2t}d_{\text{BHV}}(x_0,x)^2
$$
which is clearly measurable as a function of $(x_0,x)$ if $\nu(x_0,x)$ and $I_{(x_0,x)\in \mathcal{G}}$ are both measurable as functions of $(x_0,x)$. 
We note that the indicator function $I_{(x_0,x)\in \mathcal{G}}$ is measurable if $\mathcal{G}$ is in the product $\sigma$-algebra $\mathcal{A} \bigotimes \mathcal{A}$. 
Since $\BHV{N}^{(0)}$ is a separable metric space, the product $\sigma$-algebra and the Borel $\sigma$-algebra coincide. 
Hence to prove that $I_{(x_0,x)\in \mathcal{G}}$ is measurable, it suffices to prove that $\mathcal{G}$ is an open set.

Consider the function on $\BHV{N}^{(0)} \times \BHV{N}^{(0)}$ which for a pair $(x_0,x)$ gives the minimum distance of a point on the geodesic between $x_0,x$ from the union of codimension-2 orthants of $\BHV{N}$ (i.e. the set of trees with $\leq N-5$ internal edges). 
This union forms a closed set, and the function is well-defined since the geodesic segment is compact. 
This function is non-zero if and only if the geodesic between $x_0,x$ is simple, since in that case it avoids codimension-2 singularities. 
The function is continuous and since $\mathcal{G}$ is the preimage of an open set,  $\mathcal{G}$ is open. 
A similar argument applies to $\nu(x_0,x)$.  
\end{proof}

Next, we will show that Condition II holds in case (ii). 
We explicitly construct a set that has a different probability under the two distributions. 
The construction involves moving out along the infinite ray from the origin to $x_0$. A ball sufficiently far out along this ray has different probabilities under the two different measures.

\begin{lemma}\label{lem:cons:IdentDiffOrthants}
    Suppose $x_0 \neq x_0'$ with $d(x_0,0)=d(x_0',0)$. Suppose $x_0$ and $x_0'$ both belong to different maximal orthants $\mathcal{O}_0$ and $\mathcal{O}_1$. Then there exists an open set $A_0 \subset \mathcal{O}_0$ satisfying \begin{equation}\label{eq:cons:setsDiff}
        \BB{x_0}{t_0}(A_0) > \BB{x_0'}{t_0}(A_0).
    \end{equation}
    
\end{lemma}
\begin{proof}
We enumerate the splits in $x_0$ by $s_1,\ldots, s_{N'}$ and the splits in $x_0'$ by $u_1,\ldots, u_{N'}$. 
We define the geodesic $\gamma_0=\Gamma_{0,x_0}$ and extend it infinitely at the $x_0$ end. 
We define the sequence $(a_i)_{i \in \mathbb{N}}$ where $a_i$ is the point on $\gamma_0$ at a distance $i$ from $x_0$ in the infinite direction. 
We will find a $J \in \mathbb{N}$ and an $r >0$ such that Equation~\eqref{eq:cons:setsDiff} holds with $A_0 = K(a_J,r)$. 
For a set $A\subset\mathcal{O}_0$, we will abuse notation to additionally write $A$ for the set $\{x =(a(s_1),\ldots , a(s_{N'}))\;:\; a \in A\}\subset \mathbb{R}^{N'}_{\geq 0}$. 
We will denote the distribution of endpoints of Brownian motion on $\mathbb{R}^{N'}_{\geq 0}$ by $\BBRPlus{x_0}{t_0}$. 
We let $\Phi(x_0,t_0)$ be the probability distribution of the isotropic Gaussian with mean $x_0$ and variance $t_0$ in $\mathbb{R}^{N'}$. 
We will use the fact that both the measures $\BB{x_0}{t_0}$ and $\BBRPlus{x_0}{t_0}$ can be defined via related measures on continuous paths and split the sets of paths into those that hit a boundary and those that do not. 
We let $C$ denote the set of continuous paths $\eta :[0,t_0] \to \BHV{N}$ with $\eta(0)=x_0$ and $\eta(t_0)\in \BHV{N}^{(0)}$, that do not hit a boundary of codimension-$2$. 
We let $C_0$ be the subset of $C$ consisting of paths that hit no boundaries in the time $[0,t_0]$ and let $C_1$ be subset of $C$ containing paths that hit at least one codimension-$1$ boundary in the time $[0,t_0]$. 
Then we have $C=C_0\cup C_1$. 
We denote by $\Bpath{x_0}{t_0}$ the Brownian motion measure on $C$, where the $\sigma$-algebra was given in~\cite{nye2020random}. 
For a measurable subset $A\subset \mathcal{O}_0$ let $C(A)$ be the set of paths in $C$ that have their endpoints in $A$ and for $i=0,1$, let $C_i(A)$ be the set of paths in $C_i$ that have their endpoints in $A$. 
The projection map $\mathcal{P}$ from~\cite{nye2020random} maps elements of $C$ to paths on the positive orthant $\mathbb{R}^{N'}_{\geq 0}$. 
We let $\mathcal{P}C(A)$ be the set of projected paths that have their endpoints in $A$ when $A$ is considered as a subset of $\mathbb{R}^{N'}_{\geq 0}$, and define $\mathcal{P}C_i(A)$ analogously for $i=1,2$. 
The projection map $\mathcal{P}$ and the sets $C$, $C_0$ and $C_1$ are defined analogously for the source $x_0'$ using the notation $\mathcal{P}'$, $C'$, $C'_0$ and $C'_1$ respectively.
The following two equalities hold trivially,
    \begin{equation}\label{eq:cons:noHitEq}
   \Bpath{x_0}{t_0}(C_0(A))=\BRPluspath{x_0}{t_0}(\mathcal{P}C_0(A)),
    \end{equation}
    and $$
    \BB{x_0}{t_0}(A)=\Bpath{x_0}{t_0}(C(A))=\Bpath{x_0}{t_0}(C_0(A))+\Bpath{x_0}{t_0}(C_1(A)).$$
Since any Brownian motion path that starts at $x_0'$ and ends in $A$ must traverse at least one codimension-$1$ boundary, we also have
$$
\BB{x_0'}{t_0}(A)=\Bpath{x_0'}{t_0}(C'(A))=\Bpath{x_0'}{t_0}(C_1'(A)).
$$
Let $\mathcal{S}=\{ S \subset \{1,\ldots,N'\} : S\neq \emptyset \}$ and let $g_S:\mathbb{R}^{N'}_{\geq 0} \to \mathbb{R}^{N'}$ be defined by
\begin{equation*}
(g_S(x))_j =\begin{cases} -x_j & \text{if } j \in S,\\
x_j & \text{otherwise.}
\end{cases}
\end{equation*}
Using the reflection principle we write
$$
\BRPluspath{x_0}{t_0}(\mathcal{P}C(A))= \BBR{x_0}{t_0}(A) + \sum_{S \in \mathcal{S}}\BBR{x_0}{t_0}(g_S(A)).
$$
Using the reflection principle again, we see that 
\begin{equation}\label{eq:cons:reflectionPrinciple}
    \BRPluspath{x_0}{t_0}(\mathcal{P}C_1(A)) = 2 \sum_{S \in \mathcal{S}}\BBR{x_0}{t_0}(g_S(A)),
\end{equation}
which also gives that
$$
\BRPluspath{x_0}{t_0}(\mathcal{P}C_0(A)) = \BBR{x_0}{t_0}(A) - \sum_{S \in \mathcal{S}}\BBR{x_0}{t_0}(g_S(A)).
$$
We hence obtain the following lower bound for the probability of $A$ under $B(x_0,t_0)$,
\begin{equation}\label{eq:cons:x0Bound}
    \BB{x_0}{t_0}(A) > \BBR{x_0}{t_0}(A) - \sum_{S \in \mathcal{S}}\BBR{x_0}{t_0}(g_S(A)).
\end{equation}

Next we obtain an upper bound for $B(x_0',t_0)(A)$ in the following way.
We let $\permGroup$ be the permutation group on $\{1,\ldots,N'\}$. 
The projection onto $\mathbb{R}^{N'}_{\geq 0}$ of any path $\eta \in C'(A)$ necessarily has $\mathcal{P}'(\eta)(t_0) \in \tau(A)$ for some $\tau\in\permGroup$, where $\tau(A)=\{(a(s_{\tau(1)}),\ldots,a(s_{\tau(N')}) : a\in A\}$. 
Since $x_0'$ is not in the orthant $\mathcal{O}_0$, pathjs $\eta$ from $x_0'$ ending in $A$ must hit at least one boundary.
We therefore have
$$
\Bpath{x_0'}{t_0}(C_1'(A))\leq\sum_{\tau \in G_{N'}}\BRPluspath{x_0'}{t_0}(\mathcal{P}'C_1'(\tau(A)).
$$
As in Equation~$\eqref{eq:cons:reflectionPrinciple}$, using the reflection principle we have
$$
\BRPluspath{x_0'}{t_0}(\mathcal{P}'C_1'(\tau(A)) = 2\sum_{S \in \mathcal{S}} \BBR{x_0'}{t_0} (g_S(\tau(A))).
$$
We adopt the following notation for the sum over permutations and reflections of the set $A$,
$$
B^{\tau}(x_0',t_0)(A) = 2\sum_{\tau \in G_{N'}}\sum_{S \in \mathcal{S}} \BBR{x_0'}{t_0} (g_S(\tau(A))),
$$
which gives a bound on $\BB{x_0'}{t_0}(A)$ by
\begin{equation}\label{eq:cons:x1Bound}
    \BB{x_0'}{t_0}(A)=\Bpath{x_0'}{t_0}(C_1'(A)) \leq B^{\tau}(x_0',t_0)(A).
\end{equation}
We now we have the bounds~$\eqref{eq:cons:x0Bound}$ and~$\eqref{eq:cons:x1Bound}$ on $B(x_0,t_0)(A)$ and $B(x'_0,t_0)(A)$ for $A\subset\mathcal{O}_0$.  

Next, let $\phi(x;x_0,t_0)$ denote the density of the isotropic normal distribution on $\mathbb{R}^{N'}$ with mean $x_0$ and variance $t_0$ and define
$$
\phi_\mathcal{S} (x ; x_0,t_0) = \sum_{S \in \mathcal{S}}\BR{x_0}{t_0}{g_S(x)}.
$$
For $x \in \mathbb{R}^{N'}_{\geq 0}$, define 
\begin{equation}\label{eq:cons:BMPermBound2}
    f_{B^\tau}(x;x_0',t_0) = 2\sum_{\tau \in G_{N'}} \phi_\mathcal{S} (\tau(x) ; x_0,t_0).
\end{equation}
We will show that for any $S \in \mathcal{S}$, 
\begin{equation}\label{eq:cons:convStatement1}
\frac{\BR{x_0}{t_0}{g_S(a_i)}}{\BR{x_0}{t_0}{a_i}} \to 0 \text{ as } i \to \infty.
\end{equation}
We will also show that for any $\tau \in G_{N'}$ and any $S \in \mathcal{S}$
\begin{equation}\label{eq:cons:convStatement2}
\frac{\BR{x_0'}{t_0}{g_S(\tau(a_i))}}{\BR{x_0}{t_0}{a_i}} \to 0 \text{ as } i \to \infty,
\end{equation}
and hence that 
\begin{align}
\frac{f_{B^{\tau}}(a_i;x_0',t_0)}{\BR{x_0}{t_0}{a_i}-\phi_\mathcal{S} (a_i ; x_0,t_0)}\nonumber 
\to 0
\end{align}
as $i\to \infty$.
We can therefore choose $J \in \mathbb{N}$ such that
$$
f_{B^{\tau}}(x_0',t_0)(a_J) < \frac{1}{2} [\BR{x_0}{t_0}{a_J}-\phi_\mathcal{S} (a_J ; x_0,t_0)].
$$
By continuity of each of the functions on $\mathbb{R}^{N'}$ there is a value $r >0 $ such that
\begin{align*}
B^{\tau}(x_0',t_0)(A_0) &< \BBR{x_0}{t_0}(A_0) - \sum_{S \in \mathcal{S}}\BBR{x_0}{t_0}(g_S(A_0)) 
\end{align*}
with $A_0 = K(a_J,r)$ and $A_0\subset \mathcal{O}_0$. Recalling Equations~\eqref{eq:cons:x0Bound} and ~\eqref{eq:cons:x1Bound}, this proves the claim.

To show the convergence in Equation~\eqref{eq:cons:convStatement1} note that for $S \in \mathcal{S}$, in $\mathbb{R}^{N'}$, we have
$$
\|x_0 - g_S(a_i)\|^2-\|x_0-a_i\|^2= \sum_{j\in S} \left(2 +\frac{i}{\|x_0\|}\right)^2x_0(s_j)^2 \to \infty,
$$
as $i \to \infty$ and hence 
$$
\frac{\BR{x_0}{t_0}{g_S(a_i)}}{\BR{x_0}{t_0}{a_i}} = \exp\left(\frac{1}{2t_0}(\|x_0 - a_i\|^2-\|x_0 - g_S(a_i)\|^2)\right) \to 0
$$
as $i \to \infty$.

For the convergence in Equation~\ref{eq:cons:convStatement2}, we see that in $\mathbb{R}^{N'}$,
\begin{align}\label{eq:cons:BMPermDistBound}
&\|x_0' - g_S(\tau(a_i))\|^2 -\|x_0 - a_i\|^2\nonumber\\ = &\sum_{j \in S} (x_0'(u_j) + \tau(a_i(s_j))^2 + \sum_{j \in S^C} (x_0'(u_j) - \tau(a_i(s_j))^2-\sum_{j =1}^{N'} (x_0(s_j) - a_i(s_j))^2
\nonumber\\=& 2\sum_{j \in S} x_0'(u_j) \tau(a_i)(s_j) - 2\sum_{j \in S^C} x_0'(u_j) \tau(a_i)(s_j) +2\sum_{j =1}^{N'} x_0(s_j) a_i(s_j), 
\end{align}
using that $\sum_{j=1}^{N'} a_i(s_j)^2 =\sum_{j=1}^{N'} \tau(a_i(s_j)^2$, and $\sum_{j=1}^{N'} x_0(s_j)^2 =  \sum_{j=1}^{N'} x_0'(u_j)^2$. 
Since $\|x_0'\|=\|x_0\|$ and $a_i$ is a scalar multiple of $x_0$ as a vector in $\R^{N'}$, the Cauchy-Schwarz inequality shows that the sum of the second and third terms above is $\geq 0$.  
Since $S$ is not empty, and $a_i(s_j)\rightarrow \infty$ as $i\rightarrow\infty$,  the right hand side of Equation~\ref{eq:cons:BMPermDistBound} goes to infinity in the limit.
Since 
$$
\frac{\BR{x_0'}{t_0}{g_S(\tau(a_i))}}{\BR{x_0}{t_0} {a_i}} = \exp\left(-\frac{1}{2t_0}(\|x_0' - g_S(\tau(a_i))\|^2 -\|x_0 - a_i\|^2)\right)
$$
we therefore have the convergence in Equation~\ref{eq:cons:convStatement2}.
\end{proof}



\subsection{Additional material for the simulation study}
This section contains the following figures and tables:

\begin{itemize}
    \item Figure~\ref{fig:bridgeSimsLikelihoodPlots}: \nameref{fig:bridgeSimsLikelihoodPlots}.
    \item Table~\ref{tab:timesTakenForMLEsts}: \nameref{tab:timesTakenForMLEsts}.
    \item Figure~\ref{fig:ML:10TaxaTrees1}: \nameref{fig:ML:10TaxaTrees1}. 
    \item Figure~\ref{fig:ML:10TaxaConeTrees2}: \nameref{fig:ML:10TaxaConeTrees2}. 
    \item Figure~\ref{fig:ML:10TaxaConePath}: \nameref{fig:ML:10TaxaConePath}. 
    \item Figure~\ref{fig:distinctTopsPlot}: \nameref{fig:distinctTopsPlot}.
    \item Figure~\ref{fig:inferenceTopPropPlots}: \nameref{fig:inferenceTopPropPlots}.
    \item Table~\ref{tab:inferenceTimings}: \nameref{tab:inferenceTimings}.
    \item Table~\ref{tab:InferenceParamsAndAcceptanceRates}: \nameref{tab:InferenceParamsAndAcceptanceRates}. 
We investigated the similarity of the acceptance rate for the $t_0$ parameter across simulations. It appears that the acceptance probability of the $t_0$ proposal is dominated by the likelihood ratio, and the true value of $t_0$ can be shown to have a small effect on this when using the same value of $\sigma_0$ across simulations.
    \item Figure~\ref{fig:inferenceDispTracePlots}: 
    \nameref{fig:inferenceDispTracePlots}.
    \item Figure~\ref{fig:inferenceEdgeKDEs}:
     \nameref{fig:inferenceEdgeKDEs}.
\end{itemize}

\begin{figure}[H]
\begin{center}
\includegraphics[width=0.9\textwidth]{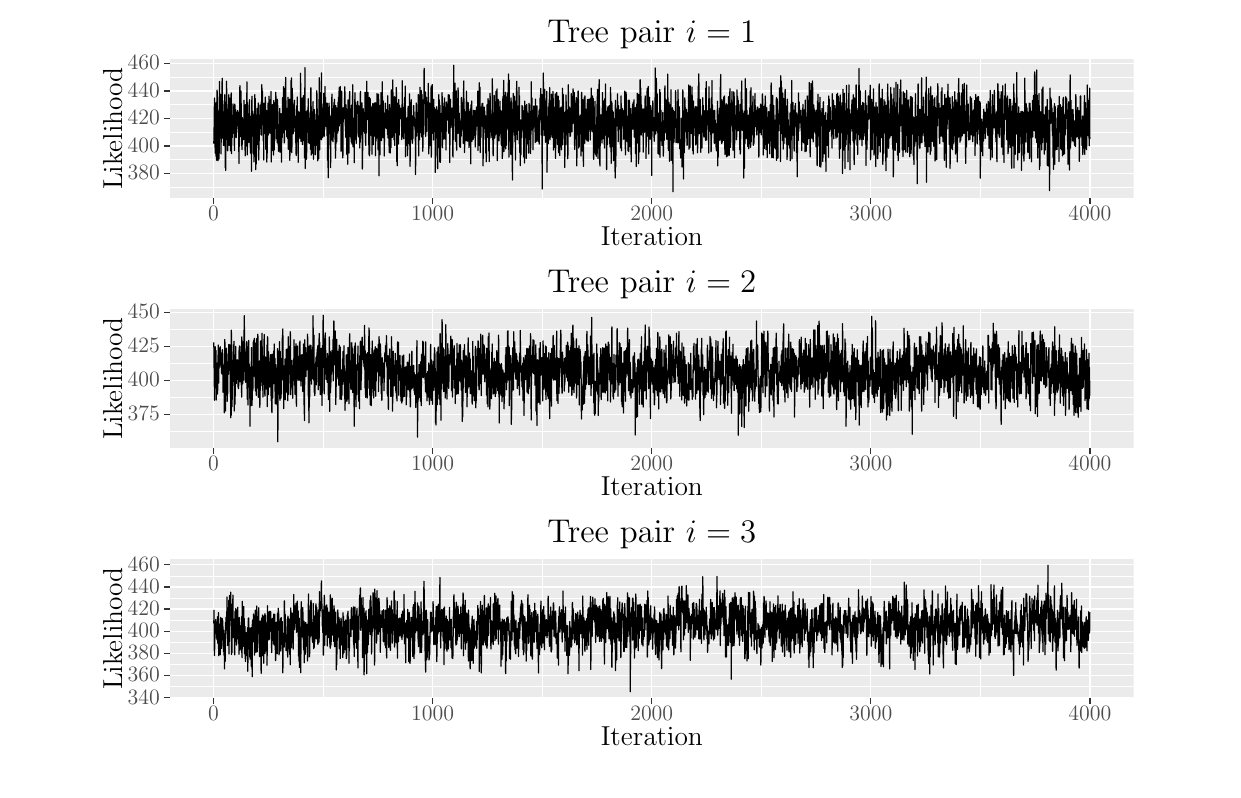}
\begin{caption}[Trace plots of the likelihood when simulating bridges between fixed endpoints in $\BHV{10}$ using Algorithm~\ref{alg:bridgesampler}.]{\label{fig:bridgeSimsLikelihoodPlots} Trace plots of the likelihood when simulating bridges between fixed endpoints in $\BHV{10}$ using Algorithm~\ref{alg:bridgesampler}.
}
\end{caption}
\end{center}
\end{figure}

\begin{center}
\begin{table}[H]
\centering
\begin{tabular}{rrrrrrr}
  \hline
   & Method & \multicolumn{3}{c}{Times taken (mins)} & Independence & Total MCMC its\\
   &  & & & & proposals & \\
   
   \hline
 & & \multicolumn{1}{|c|}{Trees $i=1$} & \multicolumn{1}{c}{Trees $i=2$} & \multicolumn{1}{|c|}{Trees $i=3$} &  & \\ 
  \hline
1 & Chib/Tunnel & 212 & 238 & 316 & $3\times 10^5$ & $1.52 \times 10^6$ \\ 
  2 & StepStone & 282 & 322 & 377 & $3\times 10^5$ & $1.22 \times 10^6$ \\ 
   \hline
   \end{tabular}
   \begin{caption}[Average time taken to obtain the samples required to compute the estimates of the marginal likelihood on $10$ taxa.]{\label{tab:timesTakenForMLEsts}Average time taken to obtain the samples required to compute the estimates of the marginal likelihood on $10$ taxa. The samples were simulated on a desktop computer with 24 2.40Ghz Intel Xeon CPUs (though all calculations were serial).
}
\end{caption}
\end{table}
\end{center}

\begin{figure}[H]
\begin{center}
\includegraphics[width=0.6\textwidth]{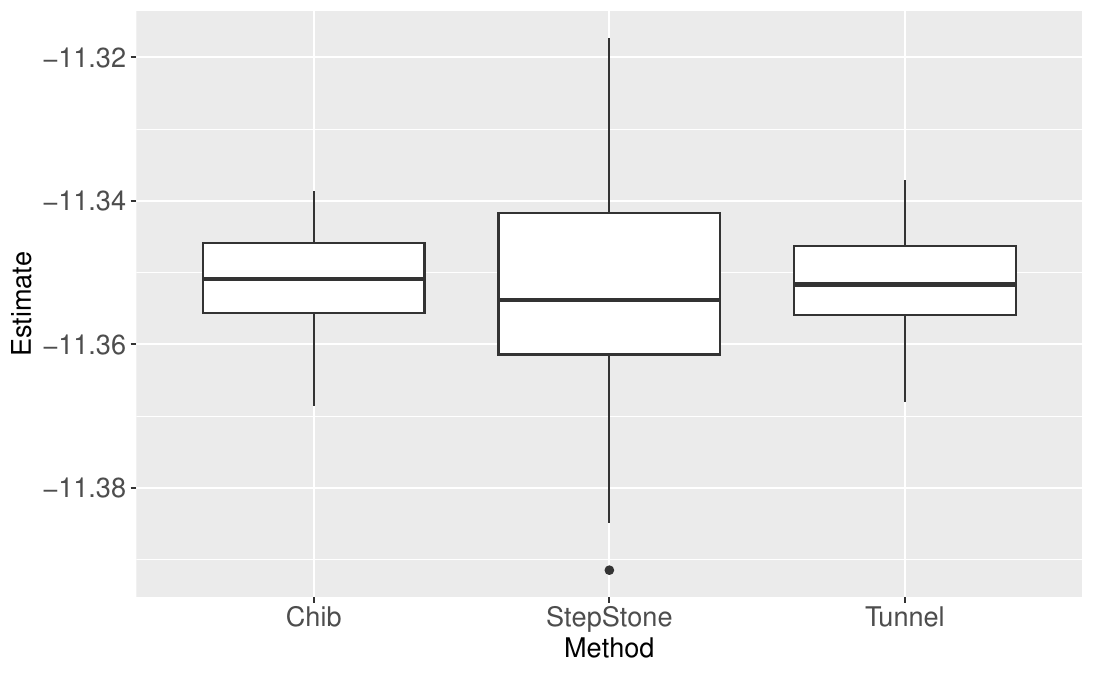}
\begin{caption}[Estimated log marginal likelihoods for tree pair $i=1$.]{\label{fig:ML:10TaxaTrees1}
Estimated log marginal likelihoods for tree pair $i=1$, $N=10$ taxa. 
The procedure was repeated $100$ times for each of the three estimators.
}
\end{caption}
\end{center}
\end{figure}

\begin{figure}[H]
\begin{center}
\includegraphics[width=0.6\textwidth]{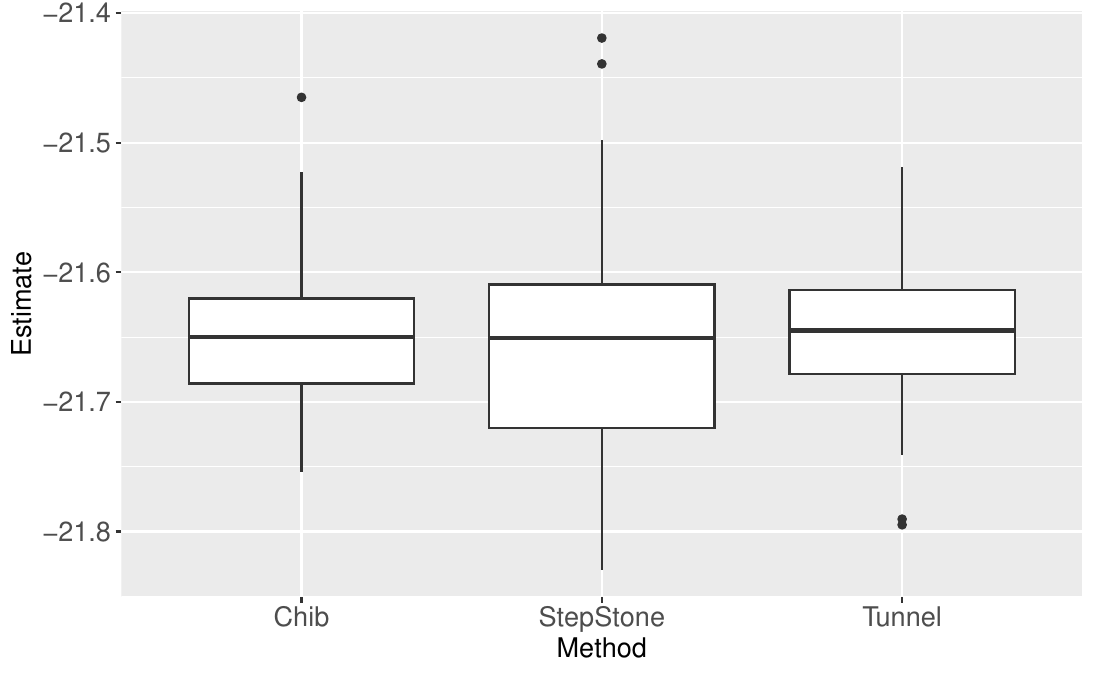}
\begin{caption}[Estimated log marginal likelihoods for tree pair $i=2$.]{\label{fig:ML:10TaxaConeTrees2}
Estimated log marginal likelihoods for tree pair $i=2$. 
}
\end{caption}
\end{center}
\end{figure}

\begin{figure}[H]
\begin{center}
\includegraphics[width=0.6\textwidth]{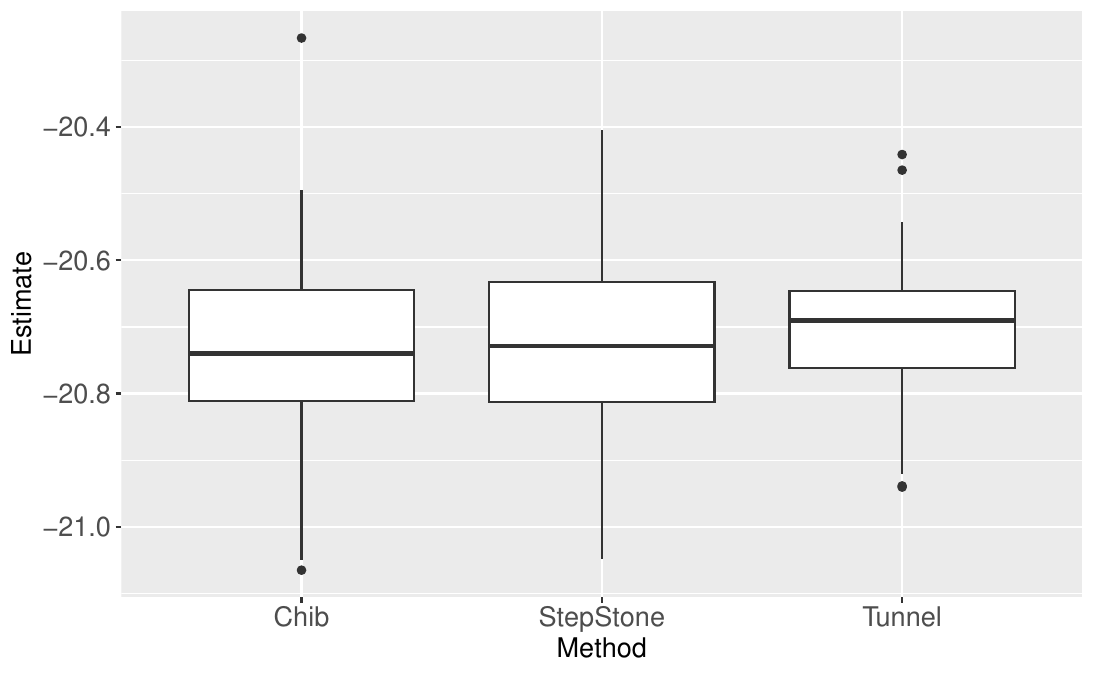}
\begin{caption}[Estimated log marginal likelihoods for tree pair $i=3$.]{\label{fig:ML:10TaxaConePath}
Estimated log marginal likelihoods for tree pair $i=3$, $N=10$ taxa. 
}
\end{caption}
\end{center}
\end{figure}

\begin{figure}[H]
\begin{center}
\includegraphics[width=0.7\textwidth]{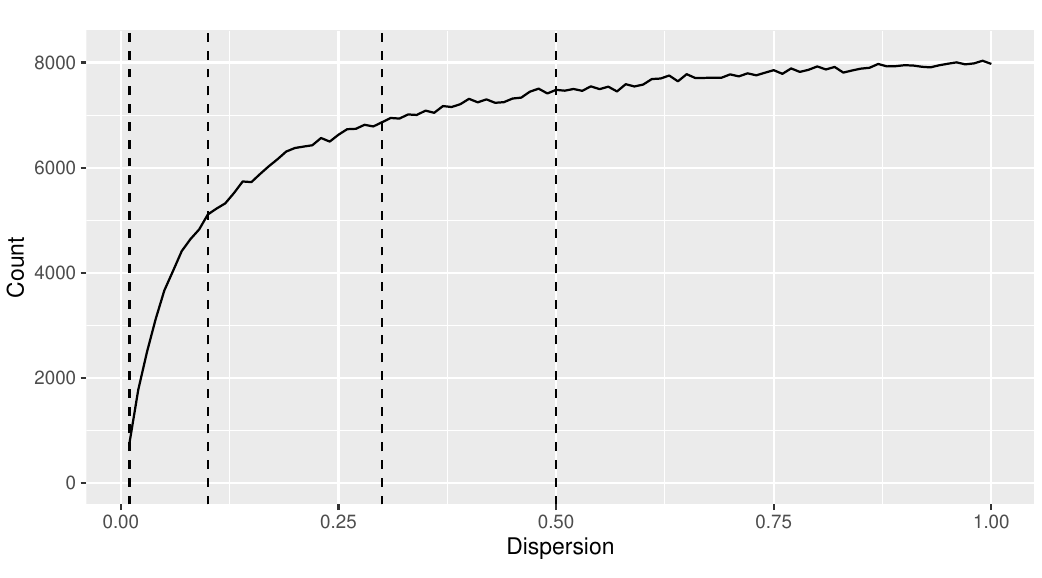}
\begin{caption}[Plot showing the count of distinct topologies obtained by forward simulating $10^4$ random walks with different values of dispersion from a fixed source tree with $N=10$ taxa and recording the topologies of the endpoints.]{\label{fig:distinctTopsPlot} Number of distinct topologies in samples of size $10^4$ from $\Wm{x_0}{t_0}$ from a fixed source tree with $N=10$ taxa and $m=2\times 10 ^3$, varying $t_0$.
}
\end{caption}
\end{center}
\end{figure}

\begin{figure}[H]
\begin{center}
\includegraphics[width=0.9\textwidth]{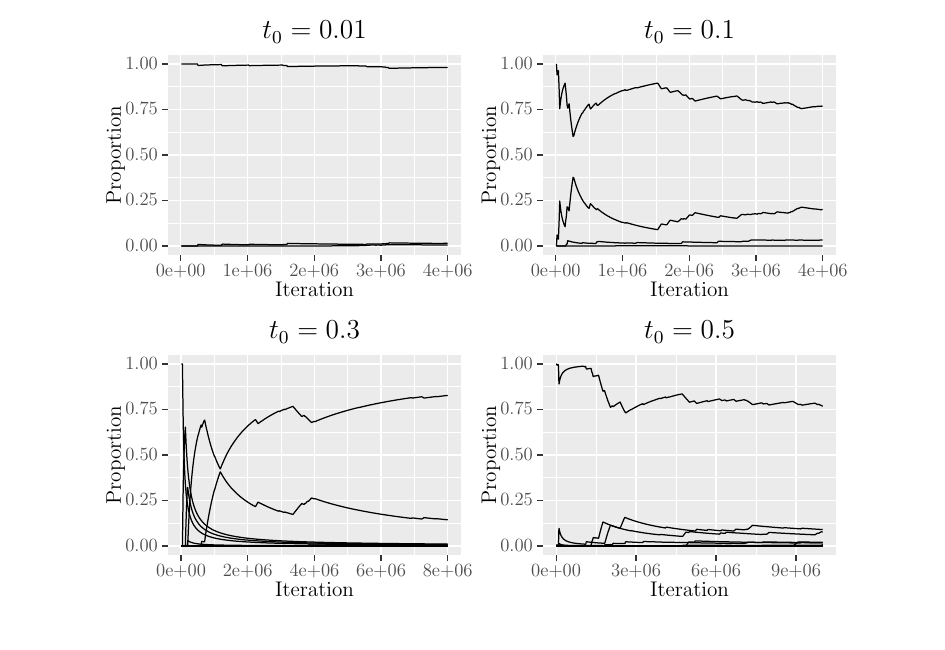}
\begin{caption}{\label{fig:inferenceTopPropPlots} Plots of the cumulative proportion of each topology that is observed in the marginal posterior sample for $x_0$ for the inferences on simulated data sets, excluding the burn-in period.
}
\end{caption}
\end{center}
\end{figure}

\begin{center}
\begin{table}[H]
\centering
\begin{tabular}{llll}
  \hline
 $t_0$ & Total iterations & Burn-in & Time taken (mins) \\ 
  \hline
 $0.01$ & $4.1 \times 10^6$ & $1.0 \times 10^5$ & $2.3 \times 10^3$ \\ 
 $0.10$ & $4.1 \times 10^6$ & $1.0 \times 10^5$ & $3.7 \times 10^3$ \\ 
 $0.30$  & $8.1 \times 10^6$ & $1.0 \times 10^5$ & $7.5 \times 10^3$ \\ 
 $0.50$  & $16.1 \times 10^6$ & $6.1 \times 10^6$ & $18.1 \times 10^3$ \\ 
   \hline
\\
\end{tabular}
\begin{caption}[The number of iterations, burn-in and computing time for each simulated data set.]{\label{tab:inferenceTimings} The number of iterations, burn-in and  the computing time for each simulated data set. The inference was performed on a desktop computer with four 3.40GHz Intel Core i7-6700 CPUs. We note that a thin of $100$ iterations was used in each inference in order to reduce the storage space needed for the output from the MCMC.
} 
\end{caption}
\end{table}
\end{center}

\begin{center}
\begin{table}[H]
\centering
\begin{tabular}{l|llll|lll}
\hline 
 & \multicolumn{4}{c|}{Parameter values}&\multicolumn{3}{c}{Acceptance rate} \\ 
  \hline
 $t_0$ & $\alpha_b$ & $\alpha_0$ & $\lambda_0$ & $\sigma_0 $ & Mean bridge & $x_0$ & $t_0$  \\ 
  \hline
 0.01 & 0.2 & 0.9 & 0.002 & 0.1 &  79.5\% & 26.2\% & 13.5\% \\ 
   0.10  & 0.2 & 0.9 & 0.002 & 0.1 & 57.2\% & 54.7\% & 13.4\% \\ 
   0.30  & 0.2 & 0.9 & 0.002& 0.1 & 53.7\% & 49.1\%& 13.4\% \\ 
   0.50  & 0.2 & 0.9 & 0.001& 0.1 & 52.2\% & 45.3\%& 13.4\% \\
   \hline
\end{tabular}
\begin{caption}{\label{tab:InferenceParamsAndAcceptanceRates}Proposal parameter values and proposal acceptance rates for the inference on simulated data sets.
}
\end{caption}
\end{table}
\end{center}

\begin{figure}[H]
\begin{center}
\includegraphics[width=0.9\textwidth]{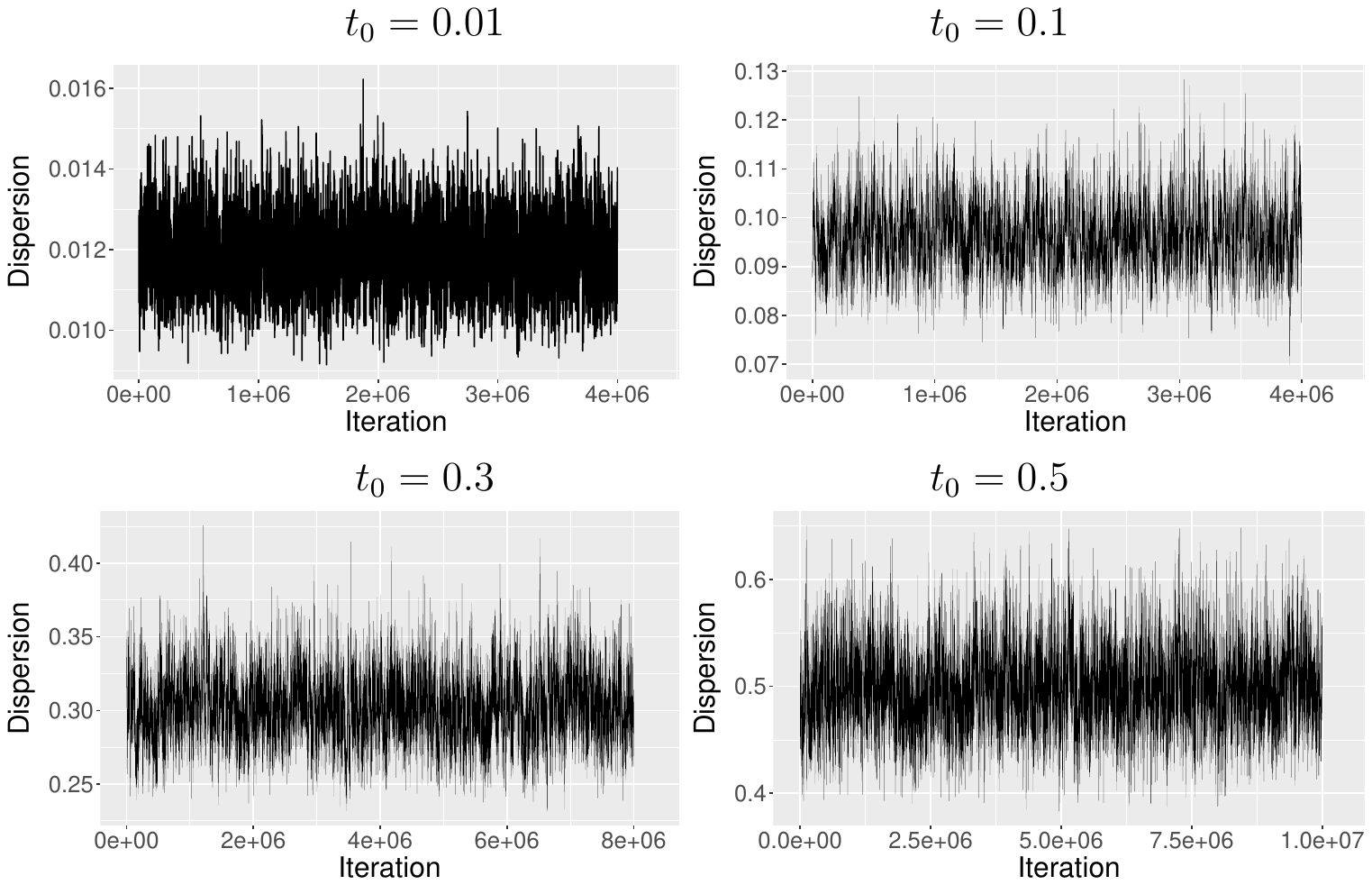}
\begin{caption}{\label{fig:inferenceDispTracePlots} 
Trace plots of the parameter $t_0$ for the inference on simulated data sets.
}
\end{caption}
\end{center}
\end{figure}

\begin{figure}[H]
\begin{center}
\includegraphics[width=0.9\textwidth]{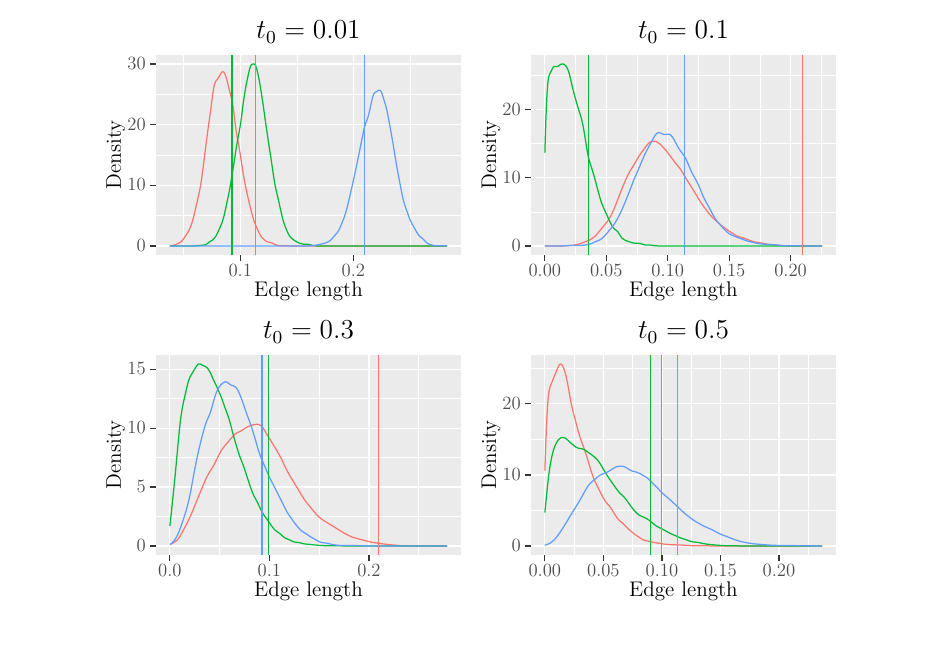}
\begin{caption}[Kernel density estimates of the lengths of three representative splits in the marginal posterior of $x_0$, conditional on the modal topology for $x_0$.]{\label{fig:inferenceEdgeKDEs} Kernel density estimates of the lengths of three representative splits in the marginal posterior of $x_0$, conditional on the modal topology for $x_0$ when the inference is tested on simulated data sets. Blue lines represent the best fitting split, whilst green lines are for an average fitting split and red lines are for the split that has the worst fit by the posterior sample. Vertical lines show the true length of the split in the source tree.
}
\end{caption}
\end{center}
\end{figure}

\subsection{Plots for the biological example}
This section contains the following figures and tables:

\begin{itemize}
    \item Figure~\ref{fig:yeastLikelihoodPlot}: \nameref{fig:yeastLikelihoodPlot}.
    \item Figure~\ref{fig:yeastDispTraceplot}: \nameref{fig:yeastDispTraceplot}.
    \item Figure~\ref{fig:yeastDispKDEplot}: \nameref{fig:yeastDispKDEplot}.
    \item Figure~\ref{fig:yeastCumulativePropPlot}: \nameref{fig:yeastCumulativePropPlot}.
    \item Figure~\ref{fig:yeastEdgeLenghtKDEplot}: \nameref{fig:yeastEdgeLenghtKDEplot}. 
    \item Figure~\ref{fig:experimentalSummaryTrees}: \nameref{fig:experimentalSummaryTrees}. 
    \item Table~\ref{tab:yeastAcceptanceRates}: \nameref{tab:yeastAcceptanceRates}. 
    \item Figure~\ref{fig:experimentalForwardSimGDs}: \nameref{fig:experimentalForwardSimGDs}. 
\end{itemize}

\begin{figure}[H]\begin{center}
\includegraphics[width=0.6\textwidth]{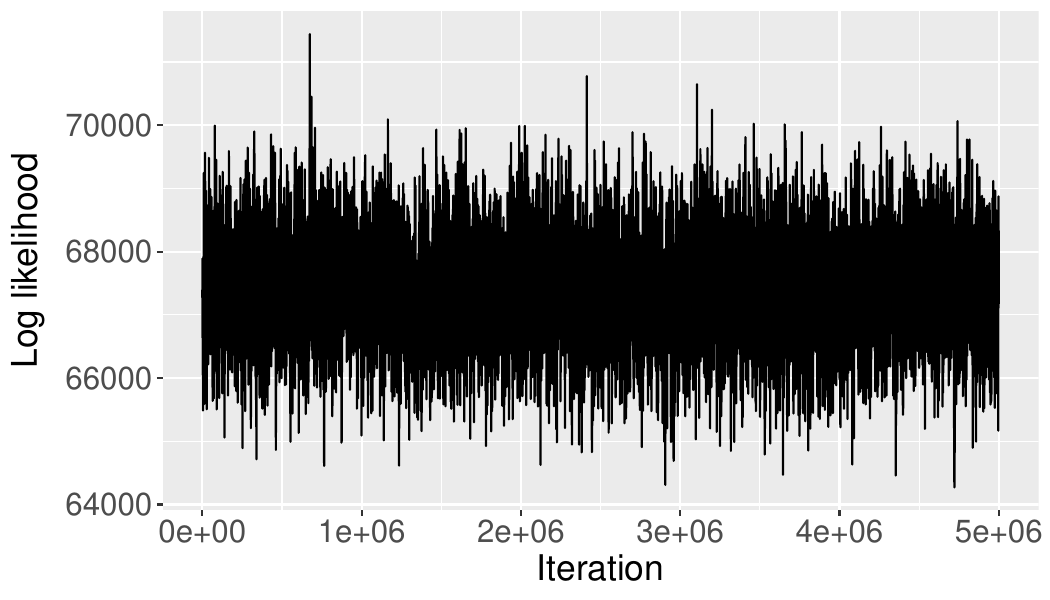}
\begin{caption}[Traceplot of the log likelihood in the inference on the data set of yeast gene trees.]{\label{fig:yeastLikelihoodPlot} 
Traceplot of the log likelihood in the inference on the data set of yeast gene trees, excluding the burn-in period.
}
\end{caption}
\end{center}
\end{figure}

\begin{figure}[H]
\begin{center}
\includegraphics[width=0.6\textwidth]{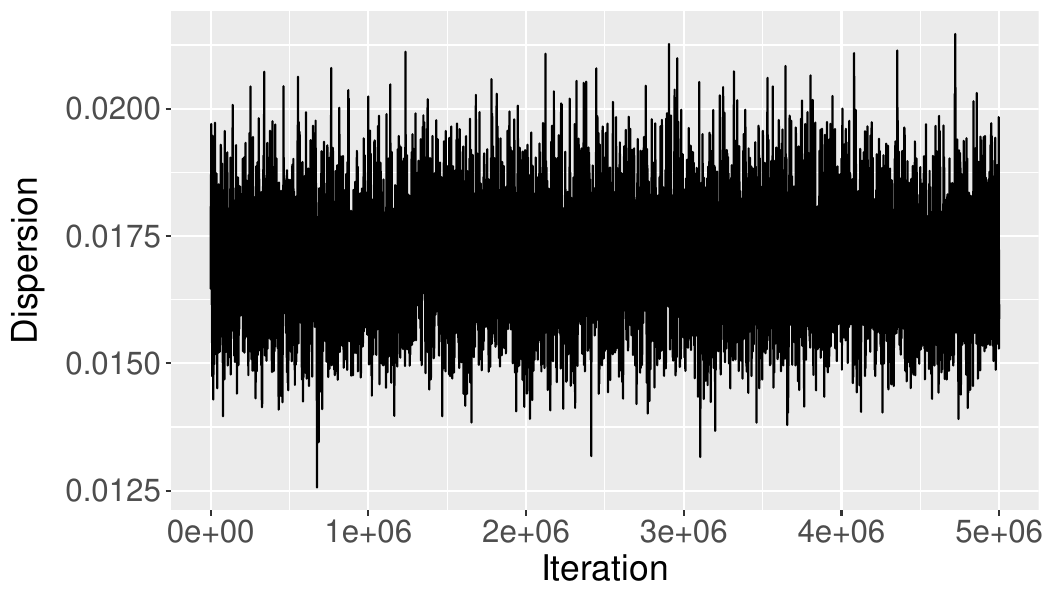}
\begin{caption}[Traceplot of the dispersion parameter in the inference on the data set of yeast gene trees.]{\label{fig:yeastDispTraceplot} 
Traceplot of the dispersion parameter in the inference on the data set of yeast gene trees, excluding the burn-in period.
}
\end{caption}
\end{center}
\end{figure}

\begin{figure}[H]
\begin{center}
\includegraphics[width=0.6\textwidth]{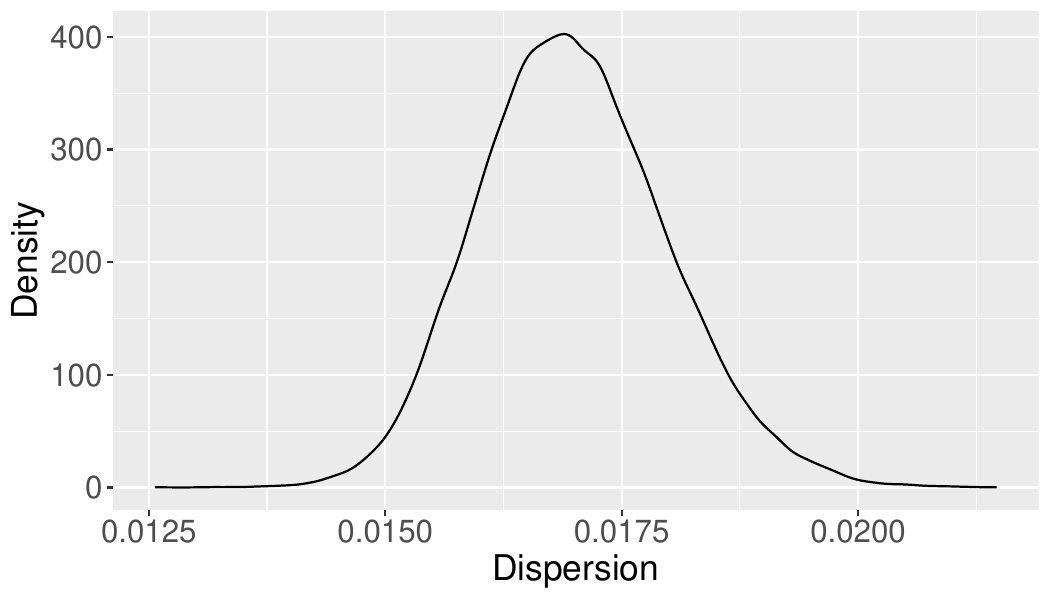}
\begin{caption}{\label{fig:yeastDispKDEplot} 
Kernel density estimate of the posterior for $t_0$ for the yeast gene trees.
}
\end{caption}
\end{center}
\end{figure}

\begin{figure}[H]
\begin{center}
\includegraphics[width=0.6\textwidth]{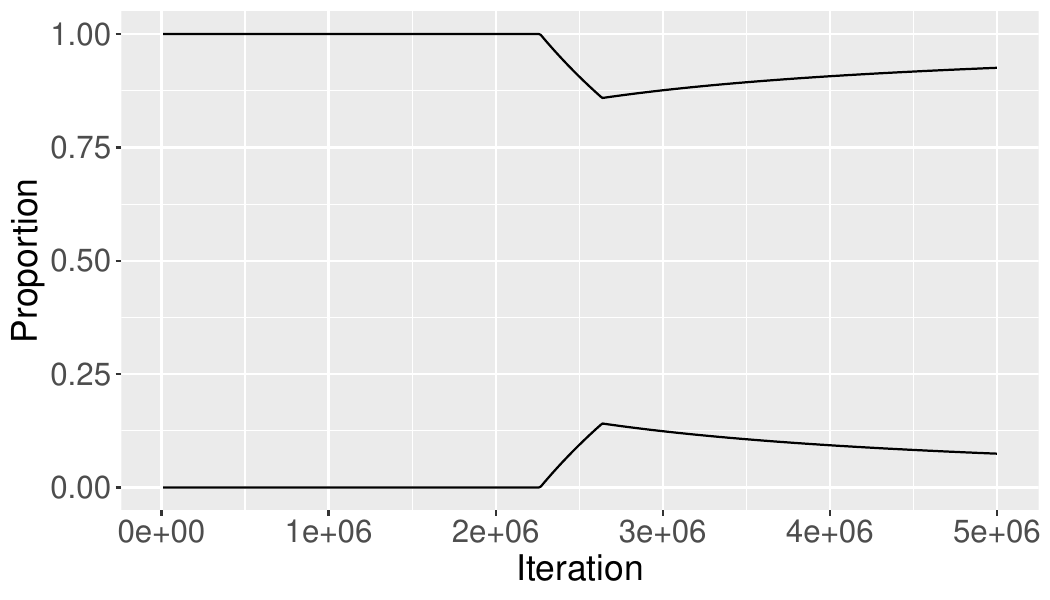}
\begin{caption}{\label{fig:yeastCumulativePropPlot} 
Plot of the cumulative proportion of the topologies observed in the posterior sample for $x_0$ for the yeast data set.
}
\end{caption}
\end{center}
\end{figure}

\begin{figure}[H]
\begin{center}
\includegraphics[width=0.6\textwidth]{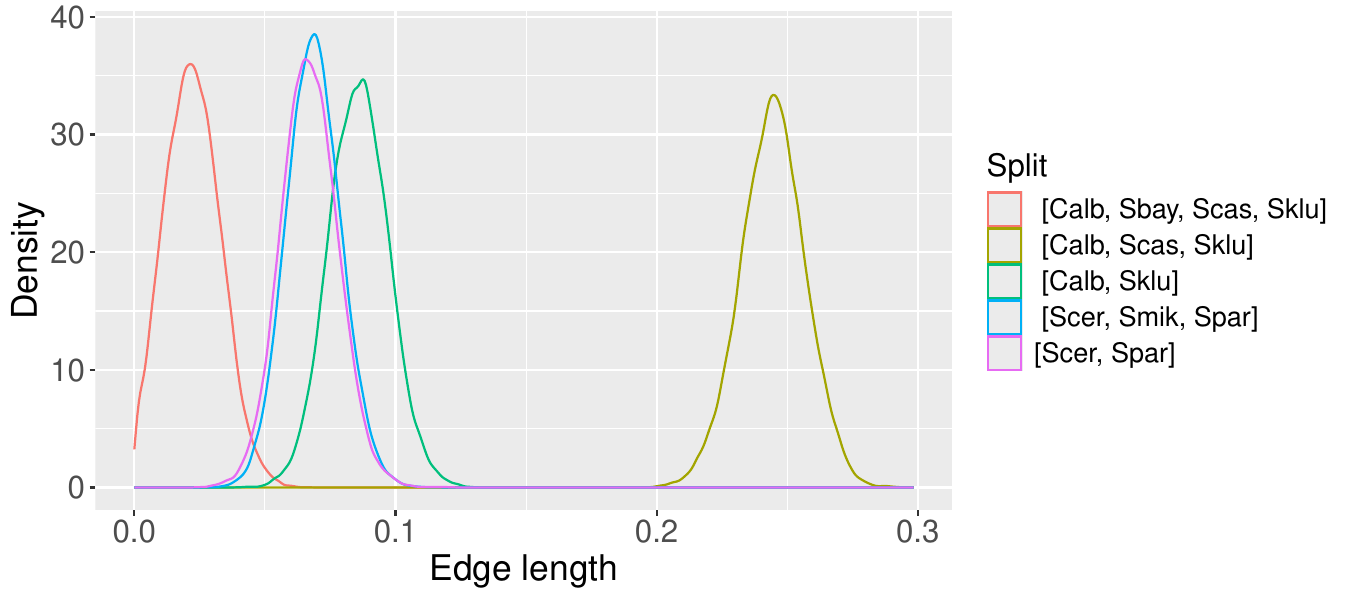}
\begin{caption}{\label{fig:yeastEdgeLenghtKDEplot} 
Plot of the kernel density estimates of the split lengths in the modal topology for the yeast data set.
}
\end{caption}
\end{center}
\end{figure}

\begin{figure}[H]
\begin{center}
\includegraphics[width=0.6\textwidth]{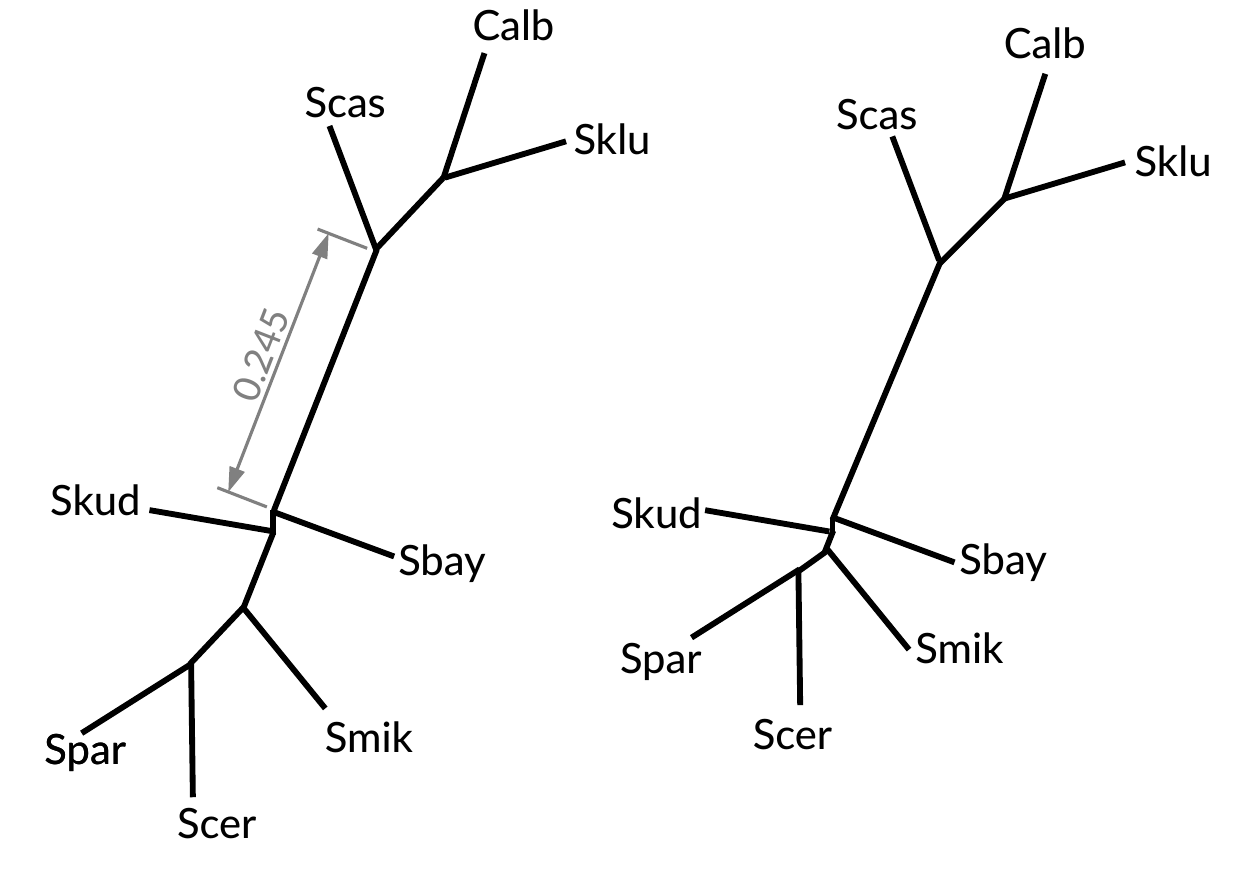}
\begin{caption}[The posterior modal source tree and the Fr\'echet mean of the yeast gene trees.]{\label{fig:experimentalSummaryTrees} The posterior modal source tree (left) and the Fr\'echet mean (right) of the yeast gene trees. 
The posterior modal source tree has the modal topology, and conditional on this, modal edge lengths.  
}
\end{caption}
\end{center}
\end{figure}

\begin{center}
\begin{table}[H]
\centering
\begin{tabular}{lllllll}
\hline 
 \multicolumn{4}{c}{Parameter values}&\multicolumn{3}{c}{Acceptance rate} \\ 
  \hline
 $\alpha_b$ & $\alpha_0$ & $\lambda_0$ & $\sigma_0 $ & Mean bridge & $x_0$ & $t_0$  \\ 
  \hline
 0.08 & 0.9 & 0.002 & 0.1 &  63.2\% & 24.8\% & 11.0\% \\ 
   \hline
\end{tabular}
\begin{caption}{\label{tab:yeastAcceptanceRates}Parameter values and proposal acceptance rates for the inference on the data set of $106$ yeast gene trees.
}
\end{caption}
\end{table}
\end{center}

\begin{figure}[H]
\begin{center}
\includegraphics[width=0.6\textwidth]{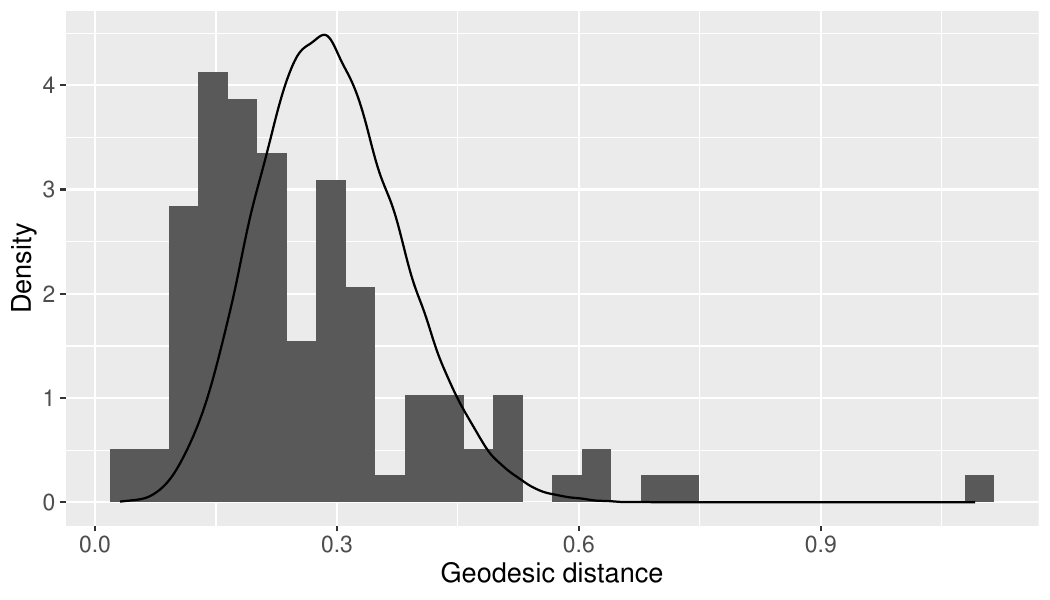}
\begin{caption}[Distribution of geodesic distances between the posterior modal source tree and (i) the data set, (ii) a the set of particles simulated under the model.]{\label{fig:experimentalForwardSimGDs} Distribution of geodesic distances between the posterior modal source tree and (i) the data set (bars), (ii) a the set of particles simulated under the model with source and dispersion parameters fixed at the posterior mode (continuous kernel density estimate plotted).
}
\end{caption}
\end{center}
\end{figure}


\end{document}